\documentclass[a4paper, 12pt, reqno]{amsart}

\numberwithin{equation}{section}

\usepackage{amssymb}
\usepackage{mathrsfs}
\usepackage{dsfont}
\usepackage{stmaryrd, MnSymbol}
\usepackage{enumerate, xspace}
\usepackage{changebar}
\usepackage{comment}
\usepackage[margin=1truein]{geometry}
\usepackage{color}
\usepackage{mathtools}
\mathtoolsset{showonlyrefs,showmanualtags}

\newtheorem{theorem}{Theorem}
\newtheorem*{theorem*}{Theorem}
\newtheorem{lemma}{Lemma}[section]
\newtheorem{corollary}{Corollary}[section]
\newtheorem{proposition}{Proposition}[section]

\newtheorem{remark}{Remark}[section]

\numberwithin{equation}{section}

\renewcommand{\a}{\alpha}
\renewcommand{\b}{\beta}
\newcommand{\e}{\varepsilon}

\newcommand{\E}{\mathbf{E}}

\def\R{{\mathbb{R}}}
\def\N{{\mathbb{N}}}
\def\Z{{\mathbb{Z}}}
\def\T{{\mathbb{T}}}
\def\E{{\mathbb{E}}}

\def\sgn{{\operatorname{sgn}}}

\begin{document}

\title{Superballistic and superdiffusive scaling limits of stochastic harmonic chains with
long-range interactions}

\author[H.~Suda]{Hayate Suda}
\address{Graduate School of Mathematical Sciences, University of Tokyo, 3-8-1, Komaba, Meguro-ku, Tokyo, 153--8914, Japan}
\email{hayates@ms.u-tokyo.ac.jp}

\begin{abstract}
We consider one-dimensional infinite chains of harmonic oscillators with random exchanges of momenta and long-range interaction potentials which have polynomial decay rate $|x|^{-\theta}, x \to \infty, \theta > 1$ where $x \in \Z$ is the interaction range. The dynamics conserve total momentum, total length and total energy. We prove that the systems evolve macroscopically  on superballistic space-time scale $(y \e^{-1}, t \e^{- \frac{\theta - 1}{2}})$ when $1 < \theta < 3$, $(y \e^{-1}, t \e^{-1} \sqrt{- \log(\e^{-1})}^{-1})$ when $\theta = 3$, and ballistic space-time scale $(y \e^{-1}, t \e^{-1})$ when $\theta > 3$. Combining our results and the results in \cite{S2}, we show the existence of two different space-time scales on which the systems evolve. In addition, we prove scaling limits of recentered normal modes of superballistic wave equations, which are analogues of Riemann invariants and capture fluctuations around characteristics. The space-time scale is superdiffusive when $2 < \theta \le 4$ and diffusive when $\theta > 4$.
\end{abstract}


\maketitle

\section{Introduction}

In recent years, one-dimensional stochastic harmonic chains have been considered as good approximations of some non-random anharmonic chains and mostly studied from the point of view of anomalous heat transport and corresponding superdiffusion of energy \cite{BBO,BOS,JKO}, see also the review \cite[Chapter 5]{Ls}. Roughly speaking, the model is defined as Hamiltonian dynamics perturbed by momentum-exchange noise which acts only on momenta locally: 
\begin{align}
\begin{dcases}
d q_{x}(t) = p_{x}(t) dt \\
d p_{x}(t) = - \sum_{z \in \Z}\a_{x-z} q_{z}(t) dt + \gamma S(p_{x-1}(t), p_{x}(t), p_{x+1}(t) ; dw_{x-1}(t), dw_{x}(t), dw_{x+1}(t)).
\end{dcases}
\end{align}
Here $p_{x}$ stands for the momentum of the particle $x \in \Z$, and $q_{x}$ for the position of the particle $x \in \Z$. Typical assumptions of interaction potential $\a_{x}, x \in \Z$ are as follows:

$(a.1) ~ \a_{x} \le 0 $ for all $x \in \Z \setminus \{ 0 \}$, $\a_{x} \neq 0$ for some $x \in \Z$.

$(a.2) ~ \a_{x} = \a_{-x} $ for all $ x \in \Z.$

$(a.3) $ There exists some positive constant $C>0$ such that $|\a_{x}| \le Ce^{-\frac{|x|}{C}} $ for all $x \in \Z$.

$(a.4) ~ \hat{\a}(k) >0 $ for all $k \neq 0$ , $\hat{\a}(0) = 0 , \hat{\a}^{''}(0) > 0$.

\noindent Here $\hat{\a}$ is the discrete Fourier transform defined as
\begin{align}
\hat{\a}(k) := \sum_{x \in \Z} \a_{x} e^{-2 \pi \sqrt{-1}x k} , \quad k \in \T, 
\end{align}
and $\displaystyle{\T = [-\frac{1}{2}, \frac{1}{2}) }$ is the one-dimensional torus. Especially, exponential-decay interaction proportional to interaction range is assumed. In addition, $\{ w_{x}(t) ; x \in \Z, t \ge 0 \}$ are i.i.d. standard Brownian motions, $\gamma > 0$ is the strength of stochastic perturbation $S$, and $S$ is added to destroy infinite number of local conservation law, but to conserve total momentum $\displaystyle{\sum_{x} p_{x}}$, total length of the system $\displaystyle{\sum_{x}r_{x} , ~ r_{x} := q_{x} - q_{x-1}}$ and total energy $\displaystyle{\sum_{x}e_{x}, ~ e_{x} := \frac{1}{2}|p_{x}|^{2} - \frac{1}{4} \sum_{z \in \Z} \a_{x-z} (q_{x} - q_{z})^{2}}$ formally. We will give rigorous definition of $S$ in Section 3. The physical meaning of $r_{x}, x \in \Z$ is the inter-particle distance or tension between neighbor particles. We call this model \textit{exponential decay model}.  

In this paper, we consider the models which have strong long-range interaction, that is, our interaction potential is given by
\begin{align}
\a_{x} := - |x|^{- \theta} , \quad x \in \Z \setminus \{ 0 \} \quad \a_{0} := 2 \sum_{x \in \N} |x|^{- \theta} \quad \theta > 1.
\end{align}
Our model does not satisfy the conditions $(a.3)$. In addition, if $\theta \le 3$ then $(a.4)$ is not satisfied because $\sum_{x \in \Z} |x|^{2 - \theta} = \infty$ and we can not define $\a^{''}$ as a continuous function on $\T$. From this polynomial decay model, we obtain following new phenomena for stochastic harmonic chains:
\begin{description}

\item[Two different space-time scales without stochastic perturbation.]  
In \cite{KO2}, the authors pointed out an interesting feature of exponential decay models: there exists two different equilibrium states and corresponding different space-time scales of the conserved quantities due to the stochastic perturbation. One is called \textit{mechanical equilibrium}, which means that the macroscopic profiles of momenta and tensions are constant, and the ohter is called \textit{thermal equilibrium}, which means that the macroscopic profile of energy is constant. At first, the system go to mechanical equilibrium at ballistic scaling. Actually, they prove the weak convergence of the scaled empirical measure of $\big(p_{x}(t), r_{x}(t), e_{x}(t) \big)$: 
\begin{align}
\lim_{\e \to 0} \e \sum_{x \in \Z} \E_{\e} [ \begin{pmatrix} p_{x}(\frac{t}{\e}) \\ r_{x}(\frac{t}{\e}) \\ e_{x}(\frac{t}{\e})  \end{pmatrix} ] \cdot \begin{pmatrix} J_1(\e x,t) \\ J_2(\e x,t) \\ J_3(\e x,t) \end{pmatrix} = \int_{\R} dy ~ \begin{pmatrix} \bar{p}(y,t) \\ \bar{r}(y,t) \\ \bar{e}(y,t) \end{pmatrix} \cdot \begin{pmatrix} J_1(y,t) \\ J_2(y,t) \\ J_3(y,t) \end{pmatrix}
\end{align}
for any test function $J_i \in C^{\infty}_{0}(\R \times [0,\infty)) , i=1,2,3$ where $(\bar{p},\bar{r})$ is the solution of the linear wave equation
\begin{align}\label{eq:eulereq}
\begin{dcases} \partial_{t} \bar{r}(y,t) = \partial_{y} \bar{p}(y,t) \\ \partial_{t} \bar{p}(y,t) = \frac{\hat{\a}^{''}(0)}{8\pi^{2}} \partial_{y} \bar{r}(y,t).\end{dcases}
\end{align}
In addition, $\displaystyle{\bar{e}(y,t) := \frac{1}{2}( \bar{p}^{2}(y,t) + \frac{\hat{\a}^{''}(0)}{8\pi^{2}} \bar{r}^{2}(y,t) ) + T(y)}$ and $T(y)$ is called the macroscopic temperature profile.
After the system reaches mechanical equilibrium, then $(p,r)$ terms are macroscopically constant, but the energy will evolve in time. At superdiffusive scaling the energy distribution converges weakly,
\begin{align}
\lim_{\e \to 0} \e \sum_{x \in \Z} \int_{0}^{\infty} dt ~ \E_{\e} [e_{x}(\frac{t}{\e^{3/2}})] J(\e x,t) = \int_{\R} dy \int_{0}^{\infty} dt ~ T(y,t) J(y,t)
\end{align}
for any $J \in C^{\infty}_{0}(\R \times [0,\infty))$ and the limit point $T(y,t)$ satisfies a $3/4$-fractional diffusion equation
\begin{align} 
\partial_{t} T(y,t) &= - \frac{(\a^{''}(0))^{3/4}}{2^{9/4}(3 \gamma)^{1/2}} (-\Delta)^{\frac{3}{4}}T(y,t) \\
T(y,0) &= T(y),
\end{align}
where $\gamma > 0$ is the strength of the stochastic noise. In their paper, the authors show those two pictures by dividing the variables $(p,r)$ into two terms $(p = p^{'} + p^{''}, r = r^{'} + r^{''})$, where $(p^{'},r^{'})$ is called \textit{thermal terms} and $(p^{''},r^{''})$ \textit{phononic terms}, and each term have different initial distributions, called \textit{thermal type} and \textit{phononic type}. A phononic type distribution means that there is no thermal energy, and a thermal type distribution means that the system is at mechanical equilibrium, see \cite[Section 2]{KO2}. The role of the stochastic perturbation is very crucial for the exponential decay model.  Actually, if there is no stochastic perturbation, which means that $\gamma = 0$, then the macroscopic energy transport is purely ballistic. 

For polynomial decay models, in \cite{S2} we show that the thermal energy evolves on a superdiffusive space-time scale and the time evolution law of the macroscopic thermal energy $W(y,t), y \in \R, t \ge 0$ is given by a fractional diffusion equation and the exponent of the fractional Laplacian depends on $\theta > 2$ : 
\begin{theorem*}[Theorem 1 in \cite{S2}]
Under some initial conditions, the following weak convergence holds for any test function $J \in C^{\infty}_{0}(\R \times [0,\infty))$ : 
\begin{align}
\lim_{\e \to 0} \e \sum_{x \in \Z} \int_{0}^{\infty}dt ~ \mathbb{E}_{\e}\Big[ e_{x}(\frac{t}{f_{\theta}(\e)}) \Big] J(\e x,t) = \int_{\R} dy \int_{0}^{\infty} dt ~ W(y,t)J(y,t),
\end{align}
where $f_{\theta}(\e)$ is the time scaling which defined as 
\begin{align}
f_{\theta}(\e) := \begin{dcases} \e^{\frac{6}{7 - \theta}} &1 < \theta < 3, \\ |h(\e)|^{3}  &\theta = 3, \\ \e^{\frac{3}{2}}  &\theta > 3, \end{dcases}
\end{align}
and $h(\cdot)$ is the inverse function of $y \mapsto \Big(\frac{y^{4}}{- \log{(y)}}\Big)^{\frac{1}{2}}$ on $[0,1)$. In addition, $W(y,t)$ is the solution of the following fractional diffusion equation:
\begin{align}
\partial_{t} W(y,t) = \begin{dcases} - C_{\theta,\gamma_{0}} (-\Delta)^{\frac{3}{7-\theta}} W(y,t) \quad &2 < \theta \le 3, \\ - C_{\theta,\gamma_{0}} (-\Delta)^{\frac{3}{4}} W(y,t) \quad &\theta > 3, \end{dcases}
\end{align}
where
$C_{\theta,\gamma_{0}}$ is a positive constant defined as
\begin{align}
C_{\theta,\gamma_{0}} := \begin{dcases} \frac{24 \csc{\big(\frac{(4-\theta)\pi}{7 - \theta}} \big)}{7 -\theta} \Big(\frac{(\theta - 1)}{24 }\Big)^{\frac{6}{7-\theta}} \gamma_{0}^{-\frac{\theta - 1}{7 - \theta}} C_{1}(\theta)^{\frac{3}{7-\theta}}  \quad &2 < \theta \le 3, \\ \frac{\sqrt{6}}{12} \gamma_{0}^{-\frac{1}{2}} C_{1}(\theta)^{\frac{3}{4}} \quad &\theta > 3. \end{dcases} 
\end{align}
and $C_{1}(\theta)$ is defined in Appendix \ref{app:asympofa} in this paper. 
\end{theorem*}

In the present paper, we consider the scaling limits for phononic terms of polynomial decay models. We explain our main results in the latter half of Introduction. We emphasize that by combining \cite[Theorem 1, 2]{S2}, Theorem \ref{thm:superballistic1} and Theorem \ref{thm:superballisticenergy} in this paper, we verify that if $2 < \theta \le 3$, then there exists two different space-time scales regardless of the existence of stochastic perturbation. Especially, the faster scaling is superballistic scaling and the other is ballistic or superdiffusive scaling.  That is, the long-range interaction decomposes the energy into two terms in the sense of space-time scale. 

Before we explain Theorem \ref{thm:superballistic1}, \ref{thm:superballisticenergy} and show the limiting equations with superballistic space-time scaling, we will introduce new variable, which is dual variable of the momentum for polynomial decay models.

\item[New conserved quantity for polynomial decay models]

For our model, $r_{x} = q_x - q_{x - 1}, x \in \Z$ is no longer the proper dual variable for $p_{x}, x \in \Z$ because $r_{x}$ only captures the tension between neighbor particles. Actually, when $\theta \le 3$, then we can not derive the limiting equations for the variables $(p_{x}, r_{x})$. This leads us to introduce the generalized tension, denoted by $l_{x}, x \in \Z$, which is formally defined as
\begin{align}
l_{x} = \mathcal{H}_{\Z} (\check{\omega} * q)_{x},
\end{align}
where $\check{\omega} : \Z \to \R$ is the inverse Fourier transform of $\omega(k) := \sqrt{\a(k)}, k \in \T$, and $\mathcal{H}_{\Z}$ is the Hilbert transform on $\Z$ defined in Section \ref{sec:notation}. We see that $\sum_{x} l_{x}$ is formally conserved, and if $\theta > 3$ then $r_{x}, l_{x}$ are macroscopically equal up to a constant multiple, see Corollary \ref{cor:weakconv}. Hence we can think that $l_{x}$ is a natural generalization of $r_{x}$. One of our results is the weak convergence of the scaled empirical measure of $(p_{x}(t), l_{x}(t))$. The time scaling $j_{\theta}(\e)$ is given by
\begin{align}\label{scaling:phononic}
j_{\theta}(\e) := \begin{dcases} \e^{\frac{\theta - 1}{2}}, \quad 1 < \theta < 3 , \\
\e (- \log{\e})^{\frac{1}{2}}, \quad \theta = 3 , \\
\e, \quad \theta > 3. \end{dcases}
\end{align} 
In Theorem \ref{thm:superballistic1} we show the scaling limit of the scaled dynamics $\{ (p_{x}(\frac{t}{j_{\theta}(\e)}), l_{x}(\frac{t}{j_{\theta}(\e)})); x \in \Z, t \ge 0 \}$ with initial distribution which satisfies the phononic type condition \eqref{ass:phononic}. Especially, Theorem \ref{thm:superballistic1} implies the weak convergence of the scaled empirical measure of $(p_{x}(\frac{t}{j_{\theta}(\e)}), l_{x}(\frac{t}{j_{\theta}(\e)}))$, that is, 
\begin{align}
\lim_{\e \to 0} \e \sum_{x \in \Z} \E_{\e}[ \begin{pmatrix} p_{x}(\frac{t}{j_{\theta}(\e)}) \\ l_{x}(\frac{t}{j_{\theta}(\e)})  \end{pmatrix}] \cdot \begin{pmatrix} J_{1}(\e x) \\ J_{2}(\e x) \end{pmatrix} = \int_{\R} dy ~ \begin{pmatrix} \bar{p}(y,t) \\ \bar{l}(y,t)  \end{pmatrix} \cdot \begin{pmatrix} J_{1}(y) \\ J_{2}(y) \end{pmatrix} \label{lim:weakconvofpl}
\end{align}
for any $t \ge 0$ and any test functions $J_{i} \in C^{\infty}_{0}(\R), i=1,2$. The macroscopic dynamics $\{ (\bar{p}(y,t),\bar{l}(y,t)) ; y \in \R, t \ge 0 \}$ is the solution of the following \textit{superballistic} wave equation
\begin{align}\label{eq:superballistic}
\begin{dcases} \partial_{t} \bar{l}(y,t) = \sqrt{C_{1}(\theta)} \mathcal{D}_{\theta} \bar{p}(y,t) \\ \partial_{t} \bar{p}(y,t) = \sqrt{C_{1}(\theta)} \mathcal{D}_{\theta} \bar{l}(y,t), \end{dcases}
\end{align}
where $C_{1}(\theta)$ is a positive constant defined in Appendix \ref{app:asympofa}, $\mathcal{D}_{\theta}$ is an operator defined as
\begin{align}
(\mathcal{D}_{\theta}f)(y) := \begin{dcases} \mathcal{H}_{\R} (- (- \Delta)^{\frac{\theta - 1}{4}}f)(y) \quad &1 < \theta \le 3, \\ f^{'}(y) \quad &\theta > 3, \end{dcases}
\end{align}
for any $f \in \mathbb{S}(\R)$, and $\mathcal{H}_{\R}$ is the Hilbert transform on $\R$. The reason why we call \eqref{eq:superballistic} the superballistic wave equation is that $\bar{p}(y,t)$ satisfies 
\begin{align}\label{eq:superballistic1}
\partial_{t}^{2} \bar{p}(y,t) = \begin{dcases} - C_{1}(\theta) (- \Delta)^{\frac{\theta - 1}{2}} \bar{p}(y,t) \quad &1 < \theta \le 3, \\
C_{1}(\theta) \Delta \bar{p}(y,t) \quad &\theta > 3. \end{dcases}
\end{align}
In addition, in Theorem \ref{thm:superballisticenergy} we prove the following scaling limit of the scaled empirical measure of the phononic energy : 
\begin{align}
\lim_{\e \to 0} \e \sum_{x \in \Z} \E_{\e}[ e_{x}(\frac{t}{j_{\theta}(\e)})] J(\e x) = \int_{\R} dy ~ \bar{e}(y,t) J(y)
\end{align}
where $\bar{e}(y,t)$ is given by
\begin{align}\label{def:longrangeenergy}
\bar{e}(y,t) := \begin{dcases} \frac{1}{2} \bar{p}^2(y,t) + \frac{1}{4} L_{\theta}(y,t)  & 1 < \theta < 3, \\ \frac{1}{2} \bar{p}^2(y,t) + \frac{1}{2} \bar{l}^2(y,t) & \theta \ge 3, \end{dcases} 
\end{align}
and $L_{\theta}(y,t), 1 < \theta < 3$ is defined as
\begin{align}\label{def:potentialenergy}
L_{\theta}(y,t) := \Big(-(-\Delta)^{\frac{\theta - 1}{2}} \big(\mathcal{D}_{\theta}^{-1}\bar{l} \big)^2 \Big)(y,t) - 2\Big(-(-\Delta)^{\frac{\theta - 1}{2}} \mathcal{D}_{\theta}^{-1}\bar{l}\Big) \mathcal{D}_{\theta}^{-1}\bar{l}(y,t).
\end{align}

One might think that the value $\theta = 3$ is the threshold and when $\theta > 3$ the macroscopic behaviors of the conserved quantities are essentially the same as those of the exponential decay models. However, as we will explain in the next paragraph, some effects of long-range interactions appears macroscopically when $3 <\theta < 4$. 

\item[Non-monotonic superdiffusive behaviors of normal modes]

We consider the fluctuation of \eqref{eq:superballistic} by considering \textit{normal modes}, which is an analogue of the Riemann invariant for \eqref{eq:superballistic}. For exponential decay models, in \cite{KO2} the authors show the diffusive fluctuations of the microscopic normal modes corresponding to the Riemann invariants of \eqref{eq:eulereq}. For polynomial decay models, we can define the normal modes of \eqref{eq:superballistic} and corresponding microscopic normal modes $f^{\pm}_{x}(t), ~ x \in \Z, t \ge 0$. In Theorem \ref{thm:flucnormal} we show that 
\begin{align}
\lim_{\e \to 0} \e \sum_{x \in \Z} \mathbb{E}_{\e}[f^{\pm}_{x}(\frac{t}{n_{\theta}(\e)})] \Big(S^{\pm}_{\theta}(\frac{t}{m_{\theta}(\e)})J\Big)(\e x) = \int_{\R} dy ~ \bar{F}^{\pm}(y,t) J(y),
\end{align} 
for any $J \in C^{\infty}_{0}(\R)$, where $(m_{\theta}(\e), n_{\theta}(\e))$ are the macroscopic time scaling and microscopic time scaling respectively defined as 
\begin{align}\label{ass:timescalingflucnormal}
\bigl( m_{\theta}(\e), n_{\theta}(\e) \bigr) := \begin{dcases} \bigl( \e^{3 - \theta}, \e^{\frac{5-\theta}{2}} \bigr) \quad &2 < \theta < 3, \\
( (\log{\e^{-1}})^{-1}, \e ( \sqrt{\log{\e^{-1}}})^{-1}) \quad &\theta = 3, \\
\bigl( \e^{\theta - 3}, \e^{\theta - 2} \bigr) \quad &3 < \theta \le 4, \\
\bigl( \e, \e^{2} \bigr) \quad &\theta > 4. \end{dcases} 
\end{align}
Moreover, $S^{\pm}_{\theta}(t), t \in \R$ is a semigroup generated by $\pm \sqrt{C_{1}(\theta)} \mathcal{D}_{\theta}$ and by using this semigroup we recenter the dynamics. For the definition of the macroscopic scaling limit $\bar{F}^{\pm}(y,t), y \in \R, t \ge 0$ and the precise statements, see Section \ref{sec:flucnormal}.  From \eqref{ass:timescalingflucnormal} it is revealed that if $2 <  \theta < 4$ then the space-time scaling is superdiffusive due to the long-range interaction.  We would like to emphasize that in the case $3 < \theta < 4$ the conserved quantities converge at ballistic scaling but their superdiffusive fluctuations are obtained, that is, the effect of long-range interactions appears macroscopically if $\theta < 4$. Another interesting behavior is the non-monotonic dependence on the exponent $\theta > 2$: If $2 < \theta < 3$, then the time scaling gets faster when $\theta$ gets bigger. But if $\theta > 3$, then the time scaling gets faster when $\theta$ gets smaller and the fastest scaling appears when $\theta = 3$. This non-monotonic dependence is caused by the first and second order of $\hat{\a}(k), k \to 0$ and the condition $n_{\theta}(\e) m_{\theta}(\e)^{-1} = j_{\theta}(\e)$, where $j_{\theta}(\e)$ is defined in \eqref{scaling:phononic}.

\end{description}

We follow the idea presented in \cite{KO2} to prove above results.  Since our dynamics is infinite volume, we can not use relative entropy method used in \cite{BEO,OVY}. But if we assume that the total energy at $t=0$ is macroscopically finite, then we can use simple techniques based on $\mathbb{L}^{2}$ bound of the configurations $(p_{x}, l_{x})$ and the Fourier transform. In paticular, we don't need any information of canonical states for our system.

One of the important open problems is to generalize some results for anharmonic chains with long-range interaction. A main difficulty is that we do not know cannonical states even for our harmonic models. As we mentioned above, for our models the variable $r_{x}, x \in \Z$ is not proper dual of $p_{x},x \in \Z$, so we can expect that canonical states for stochastic harmonic chains may be expressed by using $(p_{x},l_{x})$ terms. But the definition of $l_{x},x \in \Z$ heavily depends on the concrete form of interaction potential $\a_{x}, x \in \Z$ and the Fourier transform, if chains are anharmonic then it is quite difficult to find the dual variables of the momenta. 

Our paper is organized as follows: In Section 2 we prepare some notations. In Section 3 we give the rigorous definition of the dynamics. In Section 4 we state our main results, Theorem \ref{thm:superballistic1}, \ref{thm:superballisticenergy}, \ref{thm:flucnormal} and \ref{thm:flucnormalLLN}. Proofs of Theorem \ref{thm:superballistic1}, \ref{thm:superballisticenergy} are given in Section 5 and 6 respectively. In Section 7 we prove Theorem \ref{thm:flucnormal} and \ref{thm:flucnormalLLN}.

\section{Notations}\label{sec:notation}

Let $\R$ be the real line, $\Z$ be the set of all integers and $\T = [- \frac{1}{2}, \frac{1}{2})$ be the one-dimensional torus. Denote by $\ell^{2}(\Z)$ the space of all complex valued sequences $(f_x)_{x \in \Z}$ equipped with the norm 
\begin{align}
\| f \|_{\ell^{2}(\Z)} := \sum_{x \in \Z} |f_x|^2 .
\end{align}
Denote by $\mathbb{L}^{2}(\T)$ the space of all complex valued functions $f(k), k \in \T$ equipped with the norm 
\begin{align}
\| f \|_{\mathbb{L}^{2}(\T)} := \Big( \int_{\T} dk ~ |f(k)|^{2} \Big)^{\frac{1}{2}}.
\end{align}
Denote by $\mathbb{S}(\R)$ the Schwarz space on $\R$.

For $f, g : \Z \to \R, h \in \ell^{2}(\Z)$ we define $f * g: \Z \to \R$ and $\hat{h} \in \mathbb{L}^{2}(\T) $ as
\begin{align}
(f*g)_{x} &:= \sum_{x^{'} \in \Z} f_{x-x^{'}}g_{x^{'}} ,\\
\hat{h}(k) &:= \sum_{x \in \Z} e^{-2 \pi \sqrt{-1} k x} h_{x}.
\end{align}

For $f \in \mathbb{L}^{2}(\T)$ we define $\check{f} \in \ell^{2}(\Z)$ as
\begin{align}
\check{f}_x := \int_{\T} dk ~ e^{2 \pi \sqrt{-1} k x} f(k).
\end{align}

For $J \in \mathbb{S}(\R)$ we define $\widetilde{J} : \R \to \mathbb{C}$ as
\begin{align}
\widetilde{J}(\xi) &:= \int_{\R} dy ~ e^{-2 \pi \sqrt{-1} \xi y} J(y).
\end{align}

Denote by $\operatorname{sgn}(y), y \in \R$ the sign function defined as
\begin{align}
\operatorname{sgn}(y) := \begin{dcases} 1 & y > 0, \\ 0 & y = 0, \\ -1 & y < 0. \end{dcases}
\end{align}
We define the Hilbert transform on $\R$ and $\Z$, denoted by $\mathcal{H}_{\R}$ and $\mathcal{H}_{\Z}$ respectively, via their Fourier transforms :   
\begin{align}
\widetilde{(\mathcal{H}_{\R} J)}(\xi) &:=  -\sqrt{-1} \operatorname{sgn}(\xi) \tilde{J}(\xi), \quad J \in \mathbb{S}(\R), \\
\widehat{(\mathcal{H}_{\Z} f)}(k) &:=  -\sqrt{-1} \operatorname{sgn}(k) \hat{f}(k), \quad f \in \ell^{2}(\Z).
\end{align}

For two functions $f,g$ defined on common domain $A$, we write $f \lesssim g$ or $g \gtrsim f$ if there exists some positive constant $C > 0$ such that $f(a) \le Cg(a)$ for any $a \in A$. {If $f \lesssim g$ and $g$ is a positive constant, then we also write $\sup_{a \in A}f(a) \lesssim 1$.}

For two functions $f,g$ defined on common open subset $A \subset \R$ and a number $y_0 \in \overline{A}$, where $\overline{A}$ is the closure of $A$, we write $f(y) \sim g(y)$ as $y \to y_{0}$ or $f(y) = O(g(y))$ as $y \to y_0$ if 
\begin{align}
0 <  \varlimsup_{y \to y_0} \Big| \frac{f(y)}{g(y)} \Big| < \infty. 
\end{align}
In addtion, we write $f(y) = o(g(y))$ as $y \to y_0$ if
\begin{align} 
\varlimsup_{y \to y_0} \Big|  \frac{f(y)}{g(y)} \Big| = 0.
\end{align}

\section{The dynamics}

In this section we define harmonic chains with noise and long-range interactions. Since we analyze the system with finite total energy, it is appropriate for us to define the dynamics through the wave functions $\{ \hat{\psi}(k,t) ; k \in \T, t \ge 0 \}$ as $\mathbb{L}^{2}(\T)$ solution of the stochastic differential equation \eqref{def:dynamicsofpsi}. Then we can reconstruct the classical variables $\{ p_{x}(t), q_{x}(t) ; x \in \Z , t \ge 0 \}$ from $\{ \hat{\psi}(k,t) ; k \in \T, t \ge 0 \}$, and then we define the energy $\{ e_{x}(t) ; x \in \Z, t \ge 0\}$ and the generalized tension $\{l_{x}(t) ; x \in \Z , t \ge 0 \}$. However, it may be difficult to understand the physical meaning of the important functions such as $\hat{a}(k)$ and $R(k)$ from \eqref{def:dynamicsofpsi}. To clarify the meaning of the feature values, we first give a formal description of the dynamics in terms of $\{ p_{x}(t), q_{x}(t) ; x \in \Z , t \ge 0 \}$ in the introduction.

\subsection{Formal description of the dynamics and corresponding wave function}

We consider the Hamiltonian dynamics perturbed by momentum-exchange noise which acts only on momenta locally, that is, the dynamics is governed by the following stochastic dynamical system:
\begin{align}\label{def:dynamicsofpqz}
\begin{dcases}
d q_{x}(t) = p_{x}(t) dt \\
d p_{x}(t) = \Big( - (\a * q)_{x}(t) - \frac{\gamma}{2} (\b * p)_{x}(t) \Big) dt + \sqrt{\gamma} \sum_{z = -1, 0, 1} \Big( Y_{x+z}p_{x}(t) \Big) dw_{x+z},
\end{dcases}
\end{align}
where $\gamma \ge 0$ is the strength of the noise and $Y_{x}, x \in \Z$ are vector fields defined as
\begin{align}
Y_{x} := (p_{x} - p_{x+1})\partial_{p_{x-1}} + (p_{x+1} - p_{x-1})\partial_{p_{x}} + (p_{x-1} - p_{x})\partial_{p_{x+1}}.
\end{align} 
In addition, $\{ w_{x}(t) ; x \in \Z, t \ge 0 \} $ are i.i.d. one-dimensional standard Brownian motions, and $\b_{x}$ is defined as
\begin{align}
\b_{x} := \begin{dcases} 6 , &x = 0, \\
-2 , &x = \pm 1, \\ 
-1 , &x = \pm 2, \\ 
0 , &\text{otherwise}. \end{dcases}
\end{align}

At this time we assume that $\gamma = 0$, that is, the system is deterministic harmonic chain. Taking the Fourier transform of both sides of \eqref{def:dynamicsofpqz}, we have
\begin{align}\label{def:dynamicsofpqt}
\frac{d}{dt} \begin{pmatrix} \hat{q}(k,t) \\ \hat{p}(k,t) \end{pmatrix} = \begin{pmatrix} 0 & 1 \\ -\hat{a}(k) & 0 \end{pmatrix} \begin{pmatrix} \hat{q}(k,t) \\ \hat{p}(k,t) \end{pmatrix}.
\end{align}
The eigenvalues of the above matrix are $\pm \sqrt{-1} \omega(k),$  $\omega(k) := \sqrt{\hat{\a}(k)}$, and corresponding eigenvectors $\{ \hat{\psi}(k,t), \hat{\psi}^{*}(k,t) ; k \in \T, t \ge 0 \}$, called wave functions, are given by 
\begin{align}\label{def:eqofpsi}
\hat{\psi}(k,t) &:= \omega(k) \hat{q}(k,t) + \sqrt{-1} \hat{p}(k,t).
\end{align}
Actually, we can check that 
\begin{align}
\frac{d}{dt} \hat{\psi}(k,t) = - \sqrt{-1} \omega(k) \hat{\psi}(k,t).
\end{align}
For $\gamma > 0$, we also define wave functions $\{ \hat{\psi}(k,t) ; k \in \T, t \ge 0 \}$ as \eqref{def:eqofpsi}. Then from \eqref{def:dynamicsofpqz}, the time evolution of $\{ \hat{\psi}(k,t) ; k \in \T, t \ge 0 \}$ are given by
\begin{align}\label{def:dynamicsofpsi}
\begin{dcases}
d \hat{\psi}(k,t) &= \Big[ - \sqrt{-1} \omega(k) \hat{\psi}(k,t) dt - \gamma R(k) \Big( \hat{\psi}(k,t) - \hat{\psi}^{*}(-k,t) \Big) \Big] dt \\
& ~ + \sqrt{-1} \sqrt{\gamma} \int r(k,k^{'}) \Big( \hat{\psi}(k-k^{'},t) - \hat{\psi}^{*}(k^{'}-k,t) \Big) B(dk^{'},dt),
\end{dcases}
\end{align}
where 
\begin{align}
r(k,k^{'}) &:= 2\sin^{2}{(\pi k)}\sin{\big(2\pi(k-k^{'})\big)} + 2 \sin{(2\pi k)}\sin^{2}{\big(\pi(k-k^{'}) \big)},\label{def:sqrtscatt} \\
R(k) &:= \frac{\hat{\b}(k)}{4} = 2 \sin^{4}{(\pi k)} + \frac{3}{2} \sin^{2}{(2 \pi k)}, \label{def:meanscat}
\end{align}
and $B(dk,dt)$ is a cylindrical Wiener process on $\mathcal{L}^2(\T)$ defined as
\begin{align}
B(dk,dt) := \sum_{x \in \Z} e^{2 \pi k x} dk ~ w_{x}(dt).
\end{align}

\subsection{Rigorous definition of the dynamics and generalized tension}

Now we give the rigorous definition of our dynamics. Let $(E, \mathcal{F}, \mathbb{P})$ be a probability space and assume that the cylindrical Wiener process $B(dk,dt)$ is defined on $(E, \mathcal{F}, \mathbb{P})$. For any $T > 0$, we introduce a Banach space $\mathcal{H}_{T}$ defined as
\begin{align}
\mathcal{H}_{T} &:= \{ f : \T \times [0,T] \times \Omega \to \mathbb{C} ; \| f \|_{\mathcal{H}} := \bigr( \sup_{0 \le t \le T} \mathbb{E}[\|f(t)\|_{\mathbb{L}^{2}(\T)}^{2}] \bigl)^{\frac{1}{2}} < \infty \}.
\end{align}
Then we define $ \{ \hat{\psi}(k,t) \in \mathcal{H}_{T}  ; k \in \T , 0 \le t \le T \}$ as the solution of \eqref{def:dynamicsofpsi} with initial distribution $\mu_{0}$, where $\mu_{0}$ is an arbitrary probability measure on $\mathbb{L}^{2}(\T)$. The existence and uniqueness of the solution is proved by using classical fixed point theorem, see \cite[Appendix A]{S2}. Note that from \eqref{def:dynamicsofpsi} and It\^{o}'s formula we have the energy conservation law
\begin{align}\label{eq:energyconservation}
\int_{\T} dk ~ \mathbb{E}_{\mu_{0}}[|\hat{\psi}(k,t)|^2] = \int_{\T} dk ~ E_{\mu_{0}}[|\hat{\psi}(k)|^2],
\end{align}
where $\mathbb{E}_{\mu_0}$ is the expectation with respect to the dynamics which starts from $\mu_{0}$ and $E_{\mu_{0}}$ is the expectation with respect to $\mu_{0}$.

From $ \{ \hat{\psi}(k,t) ; k \in \T , t \ge 0 \}$, we can reconstruct the classical variables $\{ p_{x}(t), q_{x}(t); x \in \Z , t \ge 0 \}$ as 
\begin{align}
&\begin{dcases} p_{x}(t) := \int_{\T} dk ~ e^{2 \pi \sqrt{-1} k x} \hat{p}(k,t) \\ q_{x}(t) := \int_{\T} dk ~ e^{2 \pi \sqrt{-1} k x} \hat{q}(k,t), \end{dcases} \\
&\begin{dcases} \hat{p}(k,t) := \frac{1}{2 \sqrt{-1}} (\hat{\psi}(k,t) - \hat{\psi}(-k,t)^{*}) \\ \hat{q}(k,t) = \frac{1}{2 \omega(k)} (\hat{\psi}(k,t) + \hat{\psi}^{*}(-k,t)). \end{dcases} \label{def:pqfromwavek}
\end{align}
Note that $\hat{q}(k)$ may not be well-defined as an $\mathbb{L}^2(\T)$ element, because $\omega(\cdot)^{-1} \notin \mathbb{L}^{2}(\T)$. However, we do not use $q$ variable to prove our results, so we do not have to worry about the well-definedness. 

Now we introduce the dual variable of $\{ p_{x}(t) ; x \in \Z, t \ge 0 \}$, called the generalized tension $\{ l_{x}(t) ; x \in \Z , t \ge 0 \}$ which is given by
\begin{align}
l_{x}(t) &:= \int_{\T} dk ~ e^{2 \pi \sqrt{-1} k x} \hat{l}(k,t) , \\
\hat{l}(k,t) &:= \frac{\sqrt{-1} \operatorname{sgn}(k)}{2} \Big( \hat{\psi}(k,t) + \hat{\psi}(-k,t)^{*} \Big). 
\end{align}
\begin{remark}
As $r_{x}, x \in \Z$ is a function of $q_{x}, x \in \Z$, we can formally write $l_{x}, x \in \Z$ as a non-local function of $q_{x}, x \in \Z$. Actually, from \eqref{def:pqfromwavek} we have
\begin{align}
\hat{l}(k) = \sqrt{-1} \omega(k) \sgn{(k)} \hat{q}(k)
\end{align}
and thus we obtain
\begin{align}
l_{x} = \mathcal{H}_{\Z} (\check{\omega} * q)_{x}.
\end{align}
We also note that the Fourier transform of the tension $\hat{r}(k), k \in \T$ is defined by using wave functions as 
\begin{align}
\hat{r}(k) := \frac{1 - e^{- 2 \pi \sqrt{-1} k} }{2 \omega(k)} \Big( \hat{\psi}(k,t) + \hat{\psi}^{*}(-k,t) \Big),
\end{align}
and thus we have 
\begin{align}
\hat{r}(k) = \frac{1 - e^{- 2 \pi \sqrt{-1} k}}{\sqrt{-1} \operatorname{sgn}(k)\omega(k)} \hat{l}(k). \label{eq:randl}
\end{align}
\end{remark}

Next we define the energy $\{e_{x}(t) ;  x \in \Z, t \ge 0 \}$ as a function of $ \{ \hat{\psi}(k,t) ; k \in \T , t \ge 0 \}$. For harmonic chains, $e_{x}(t)$ is usually defined as
\begin{align}
e_{x}(t) = \frac{1}{2} p_{x}^{2}(t) - \frac{1}{4} \sum_{z \in \Z} \a_{x-z}\big(q_{x}(t)- q_{z}(t) \big)^{2}.
\end{align}
As mentioned above, $\{ q_{x}(t) ; x \in \Z, t \ge 0 \}$ may not be well-defined, but we can define $\sum_{z \in \Z} \a_{x-z}(q_{x}(t)- q_{z}(t))^{2}$ directly as a function of $ \{ \hat{\psi}(k,t) ; k \in \T , t \ge 0 \}$ by using the following argument: suppose that $\{ q_{x} ; x \in \Z \}$ is an $\ell^2(\Z)$ element, then the Fourier transform of $\sum_{z \in \Z} \a_{x-z}(q_{x}- q_{z})^{2}$ is equal to
\begin{align}
&\int_{\T} dk^{'} (\hat{\a}(k) - \a(k - k^{'}) - \a(k^{'})) \hat{q}(k^{'}) \hat{q}(k -k^{'}) \\
&=\frac{1}{4} \int_{\T} dk^{'} F(k-k^{'},k^{'}) \Big(\hat{\psi}(k^{'}) + \hat{\psi}(-k^{'})^{*}\Big) \Big(\hat{\psi}(k-k^{'}) + \hat{\psi}(k^{'}-k)^{*}\Big)
\end{align}
where 
\begin{align}
F(k,k^{'}) := \frac{\hat{\a}(k + k^{'}) - \hat{\a}(k) - \hat{\a}(k^{'}) }{\omega(k)\omega(k^{'})}.
\end{align}
Therefore we can define $\displaystyle{\sum_{z \in \Z} \a_{x-z}(q_{x}- q_{z})^{2}}$ as  the Fourier coefficient of the above integration, that is,
\begin{align}
\sum_{z \in \Z} \a_{x-z}(q_{x}- q_{z})^{2} := \frac{1}{4} \int_{\T^2} & dk dk^{'} ~  e^{2 \pi \sqrt{-1} k x} F(k-k^{'},k^{'}) \\
&\times \Big(\hat{\psi}(k^{'}) + \hat{\psi}(-k^{'})^{*}\Big) \Big(\hat{\psi}(k-k^{'}) + \hat{\psi}(k^{'}-k)^{*}\Big),
\end{align}
and thus the energy of the particle $x \in \Z$ is given by
\begin{align}
e_{x}(t) := \frac{1}{2} p_{x}^{2}(t) - \frac{1}{16} & \int_{\T^2}  dk dk^{'} ~ e^{2 \pi \sqrt{-1} k x} F(k-k^{'},k^{'}) \\
&\times \Big(\hat{\psi}(k^{'},t) + \hat{\psi}(-k^{'},t)^{*}\Big) \Big(\hat{\psi}(k-k^{'},t) + \hat{\psi}(k^{'}-k,t)^{*}\Big). \label{def:energybywave}
\end{align}
Note that from Lemma \ref{lem:asympofa} and \ref{lem:asympofaco} we see that $F(k,k')$ is uniformly bounded, and thus $e_{x}(t)$ is well-defined for any $t \ge 0$.

\section{Main Results}

\subsection{Superballistic scaling limit}

Let $(\mu_{\e})_{0 < \e < 1}$ be a family of probability measures on $\mathbb{L}^{2}(\T)$. We define $ \{ \hat{\psi}(k,t) = \hat{\psi}_{\e}(k,t) ; k \in \T , t \in \R_{\ge 0} \}$ as the solution of \eqref{def:dynamicsofpsi} with initial condition $\mu_{\e}$. Denote by $\mathbb{E}_{\e}$ the expectation with respect to the dynamics which starts from $\mu_{\e}$. We assume that $\{ ( p_{x}(t), l_{x}(t) ) = ( p_{x}(\hat{\psi})(t), l_{x}(\hat{\psi})(t) ) ; x \in \Z, t \ge 0 \}$ satisfies the following initial condition, called \textit{phononic (or mechanical) type condition} (cf. \cite[Definition 2.3]{KO2}): there exists some $\bar{p}_{0} , \bar{l}_{0} \in C^{\infty}_{0}(\R)$ such that 
\begin{align}\label{ass:phononic}
\lim_{\e \to 0} \e \sum_{x \in \Z} E_{\mu_{\e}} \Big[|p_{x} - \bar{p}_{0}(\e x)|^{2} + | l_{x} - \bar{l}_{0}(\e x)|^{2} \Big] = 0.
\end{align}
Note that \eqref{ass:phononic} is equivalent to the following condition
\begin{align}\label{ass:equiassofthm3}
&\lim_{\e \to 0} \int_{\frac{\T}{\e}} dk ~ \mathbb{E}_{\e}\Big[|\e \hat{p}(\e k,0) - \tilde{\bar{p}}_{0}(k)|^{2} + |\e \hat{l}(\e k,0) - \tilde{\bar{l}}_{0}(k)|^{2} \Big] = 0.
\end{align}
We check the equivalence in Appendix \ref{app:initialequi}. In addition, from \eqref{eq:energyconservation} and \eqref{ass:phononic} we have
\begin{align}
\int_{\T} dk ~ \mathbb{E}_{\e}[|\hat{p}(k,t)|^2 + |\hat{l}(k,t)|^2] = \int_{\T} dk ~ E_{\mu_\e}[|\hat{p}(k)|^2 + |\hat{l}(k)|^2],
\end{align}
and
\begin{align}
\varlimsup_{\e \to 0} \int_{\T} dk ~ \e \mathbb{E}_{\e}[|\hat{p}(k,t)|^2 + |\hat{l}(k,t)|^2] = \int_{\R} d\xi ~ |\tilde{\bar{p}}_{0}(\xi)|^2 +  |\tilde{\bar{l}}_{0}(\xi)|^2.
\end{align}

Now we state one of our main results. Denote by $\big(\bar{p}(y,t),\bar{l}(y,t) \big), y \in \R, t \ge 0$ the solution of \eqref{eq:superballistic} with initial condition $\big(\bar{p}_{0}(y), \bar{l}_{0}(y) \big), y \in \R$. Let us recall that the time scaling $j_{\theta}(\e)$ is defined in \eqref{scaling:phononic}.

\begin{theorem}\label{thm:superballistic1}
Suppose that $\theta > 1$, $\gamma \ge 0$ and \eqref{ass:phononic}.  For any $t \ge 0$, we have
\begin{align}
\lim_{\e \to 0} \int_{\frac{\T}{\e}} dk ~ \mathbb{E}_{\e} \Big[ |\e \hat{p}(\e k,\frac{t}{j_{\theta}(\e)}) - \tilde{\bar{p}}(k,t)|^2 + |\e \hat{l}(\e k,\frac{t}{j_{\theta}(\e)}) - \tilde{\bar{l}}(k,t)|^2  \Big] = 0,
\end{align}
and consequently we have
\begin{align}
\lim_{\e \to 0} \e \sum_{x \in \Z} \mathbb{E}_{\e} \Big[| p_{x}(\frac{t}{j_{\theta}(\e)}) - \bar{p}(\e x,t) |^{2} + | l_{x}(\frac{t}{j_{\theta}(\e)}) - \bar{l}(\e x,t) |^{2} \Big] = 0.
\end{align}
\end{theorem}

From Theorem \ref{thm:superballistic1}, \eqref{eq:randl} and Lemma \ref{lem:asympofa} we obtain the weak convergence of the empirical measure and the scaling limit of the tension $r_x(t), x \in \Z, t \ge 0$. In addition, we see that if $\theta > 3$ then $r_x$ and $l_x$ coincide macroscopically.

\begin{corollary}\label{cor:weakconv}
Suppose that $\theta > 1, \gamma \ge 0$ and \eqref{ass:phononic}. For any $t \ge 0$ and $J_{i} \in C^{\infty}_0(\R), i =1,2$ we have \eqref{lim:weakconvofpl}. In addition, we have 
\begin{align}
&\varlimsup_{\e \to 0} \int_{\frac{\T}{\e}} dk ~ \mathbb{E}_{\e} \Big[ |\e \hat{r}(\e k,\frac{t}{j_{\theta}(\e)})|^{2}  \Big] = 0 & & 1 < \theta
\le 3, \\
&\lim_{\e \to 0} \int_{\frac{\T}{\e}} dk ~ \mathbb{E}_{\e} \Big[ |\e \hat{r}(\e k,\frac{t}{j_{\theta}(\e)}) - \frac{\e}{\sqrt{C_{1}(\theta)}}  \hat{l}(\e k,\frac{t}{j_{\theta}(\e)})|^2  \Big] = 0 & & \theta > 3.
\end{align}
\end{corollary}

\begin{remark}
The fractional Laplacians and the Hilbert transform are Fourier multipliers:
\begin{align}
\widetilde{\Big(- (- \Delta)^{a} f \Big)}(\xi) &= - |2 \pi \xi|^{2a} \tilde{f}(\xi) \quad 0 < a \le 1, \\
\widetilde{(\mathcal{H}_{\R} f)}(\xi) &=  -\sqrt{-1} \operatorname{sgn}(\xi) \tilde{f}(\xi).
\end{align}
By using the above properties we easily see that 
\begin{align}
&\begin{dcases} \partial_t \tilde{\bar{p}}(\xi,t) =\sqrt{C_{1}(\theta)} \sqrt{-1} \operatorname{sgn}(\xi) |2 \pi \xi|^{\frac{\theta - 1}{2}} \tilde{\bar{l}}(\xi,t) \\
\partial_t \tilde{\bar{l}}(\xi,t) = \sqrt{C_{1}(\theta)} \sqrt{-1} \operatorname{sgn}(\xi) |2 \pi \xi|^{\frac{\theta - 1}{2}} \tilde{\bar{p}}(\xi,t) \end{dcases} & &1 <  \theta \le 3, \\
&\begin{dcases} \partial_t \tilde{\bar{p}}(\xi,t) = 2 \pi \sqrt{C_{1}(\theta)} \sqrt{-1} \xi \tilde{\bar{l}}(\xi,t) \\
\partial_t \tilde{\bar{l}}(\xi,t) =  2 \pi \sqrt{C_{1}(\theta)} \sqrt{-1} \xi \tilde{\bar{p}}(\xi,t) \end{dcases} & & \theta > 3,
\end{align}
and
\begin{align}\label{eq:superballistic2}
\partial_t^2 \tilde{\bar{p}}(\xi,t) = \begin{dcases} - C_{1}(\theta) |2 \pi \xi|^{\theta - 1} \tilde{\bar{p}}(\xi,t) & 1 < \theta \le 3, \\
- C_{1}(\theta) |2 \pi \xi|^{2} \tilde{\bar{p}}(\xi,t) & \theta > 3. \end{dcases}
\end{align}
Note that from \eqref{eq:superballistic2} we can derive \eqref{eq:superballistic1} by using the inverse Fourier transform.
\end{remark}

Next we consider the scaling limit of the energy. Recall that $\bar{e}(y,t)$ is defined in \eqref{def:longrangeenergy}.

\begin{theorem}\label{thm:superballisticenergy}
Suppose that $\theta > 1, \gamma \ge 0$ and \eqref{ass:phononic}. For any $t \ge 0$ and $J \in C^{\infty}_0(\R)$, we have 
\begin{align}\label{lim:superballisticenergy}
\lim_{\e \to 0} \e \sum_{x \in \Z} \mathbb{E}_{\e}\Big[e_{x}(\frac{t}{j_{\theta}(\e)}) \Big] J(\e x)  = \int_{\R} dy ~ \bar{e}(y,t) J(y).
\end{align} 
\end{theorem}

\begin{remark}
From Theorem \ref{thm:superballisticenergy} we see that when $\theta \ge 3$ the quantities $e_x(t)$ and $|\psi_x(t)|^2$ coincide macroscopically, that is, for any $t \ge 0$ and $J \in C^{\infty}_0(\R)$ we have
\begin{align}
\varlimsup_{\e \to 0} \Big| \e \sum_{x \in \Z} \mathbb{E}_{\e} \Big[(e_x(\frac{t}{j_{\theta}(\e)}) - |\psi_x(\frac{t}{j_{\theta}(\e)})|^2 ) \Big]J(\e x) \Big| = 0. \label{lim:equienergywave}
\end{align}
However, if $1 <  \theta < 3$, then \eqref{lim:equienergywave} does not hold. On the other hand, in \cite[Proposition 4.1]{S2}, we show that under the thermal type condition, an analogue of \eqref{lim:equienergywave} holds for $\theta > 2$. This difference is caused by initial conditions and from the above we see that $|\psi_x(t)|^2$ may not coincide the phononic energy macroscopically.
\end{remark}

\subsection{Fluctuations of normal modes}\label{sec:flucnormal}

In this subsection we consider the normal modes of the system \eqref{eq:superballistic} $f^{\iota}(y,t), ~ \iota = \pm$, an analogue of the Riemann invariants, which are defined as
\begin{align}
f^{\pm}(y,t) &:= \bar{p}(y,t) \pm \bar{l}(y,t).
\end{align}    
The time evolution of the normal modes is given by
\begin{align}
\partial_{t} f(y,t) = \pm \sqrt{C_{1}(\theta)} (\mathcal{D}_{\theta} f^{\pm})(y,t),
\end{align}
that is, $f^{\pm}(y,t)$ are the eigenvectors and $\pm \sqrt{C_{1}(\theta)} \mathcal{D}_{\theta}$ are the eigenvalues of the system \eqref{eq:superballistic}. 

For the sake of clarity, at first we consider the case $\theta > 4$. Then the characteristics of \eqref{eq:superballistic} are given by straight lines $y \pm \sqrt{C_{1}(\theta)} t$ and thus we have 
\begin{align}
f^{\pm}(y,t) = f^{\pm}_{0}(y \pm \sqrt{C_{1}(\theta)} t)
\end{align}
where $f^{\pm}_{0} := f^{\pm}(y,0)$. This observation motivates us to study fluctuations of \textit{microscopic} normal modes around \textit{macroscopic} characteristics. Considering the above, we introduce the microscopic normal modes $f^{\iota}_{x}(t), ~ x \in \Z, t \ge 0, \iota = \pm$ defined as
\begin{align}
f^{\pm}_{x}(t) := p_{x}(t) \pm l_{x}(t).
\end{align}
From Theorem \ref{thm:superballistic1} we see that 
\begin{align}
\lim_{\e \to 0} \e \sum_{x \in \Z} \mathbb{E}_{\e} [f^{\pm}_{x}(\frac{t}{\e})] J(\e x) = \int_{\R} dy ~ f^{\pm}(y,t) J(y),
\end{align}
for any $\theta > 1$, $t \ge 0$ and $J \in C^{\infty}_{0}(\R)$. To obtain the fluctuations around the characteristics at the diffusive time scaling, we need to recenter the dynamics by shifting the origin to $\pm \sqrt{C_{1}(\theta)} (t/\e)$ because the space is scaled by $\e$. Then we obtain the following scaling limit, which is the special case of Theorem \ref{thm:flucnormal} stated in the end of this subsection.
\begin{corollary}
Assume that $\theta > 4$, $\gamma > 0$ and \eqref{ass:phononic}. Then for any $t \ge 0$ and $J \in C^{\infty}_{0}$ we have
\begin{align}
&\lim_{\e \to 0} \e \sum_{x \in \Z} \mathbb{E}_{\e} [f^{\pm}_{x}(\frac{t}{\e^2})] J(\e x \mp \sqrt{C_{1}(\theta)} \frac{t}{\e}) \\ 
&= \int_{\R} dy \int_{\R} d\xi ~ e^{2 \pi \sqrt{-1} \xi y} e^{- \frac{3\gamma}{2} \xi^2 t} \Big( \tilde{f}^{+}_{0} + \tilde{f}^{-}_{0} \Big)(\xi) J(y). \label{lim:normalwhen4}
\end{align}
\end{corollary}
\begin{remark}
The above result is essentially the same as that of \cite[Theorem 3.4]{KO2}. Note that the ballistic transport of the phononic terms with diffusive fluctuations and $3/4$-superdiffusion of thermal energy agree with the theoretical prediction of \cite{S}. We also note that for polynomial decay models $3/4$-superdiffusion of thermal energy is proved in \cite{S2} when $\theta > 3$. 
\end{remark}
Now we generalize the above argument to the case $2 < \theta \le 4$. Let $\{ S^{\iota}_{\theta}(t) ; t \in \R, \iota = \pm, \theta > 1\}$ be the semigroups corresponding to $\pm \sqrt{C_{1}(\theta)} \mathcal{D}_{\theta}$ defined via its Fourier transform:
\begin{align}
\widetilde{(S^{\pm}_{\theta}(t) g)}(\xi) := \begin{dcases} e^{\pm \sqrt{C_{1}(\theta)} \sqrt{-1} \sgn(\xi) |2\pi \xi|^{\frac{\theta - 1}{2}} t} \tilde{g}(\xi) & 1 < \theta \le 3, \\
e^{\pm \sqrt{C_{1}(\theta)} \sqrt{-1} (2 \pi \xi) t} \tilde{g}(\xi) & \theta > 3 \end{dcases}
\end{align}
for any $g \in \mathbb{S}(\R)$. Note that the right-hand side of \eqref{lim:normalwhen4} can be rewritten as
\begin{align}
\e \sum_{x \in \Z} \mathbb{E}_{\e} [f^{\pm}_{x}(\frac{t}{\e^2})] \Big(S^{\pm}_{\theta}(\frac{t}{\e})J \Big)(\e x). 
\end{align}
Therefore by using these semigroups we can recenter the dynamics for any $\theta > 1$, but the correct time scaling is not diffusive if $\theta \le 4$. Before we state the scaling limits of the normal modes when $2 < \theta \le 4$, we introduce some functions. Define $\bar{F}^{\pm}(y,t), (y,t) \in \R \times \R_{\ge 0}$ as 
\begin{align}
\begin{pmatrix}\bar{F}^{+}\\ \bar{F}^{-}\end{pmatrix}(y,t) &:= \int_{\R} d\xi ~ e^{2 \pi \xi y} \exp{\{M_{\theta}(\xi)t\}} \begin{pmatrix}\tilde{\bar{F}}^{+}\\ \tilde{\bar{F}}^{-}\end{pmatrix}(\xi,0), \\
\bar{F}^{\pm}(y,0) &:= \bar{p}_{0}(y) \pm \bar{l}_{0}(y), \quad M_{\theta}(\xi) := (M^{(i,j)}_{\theta}(\xi))_{i,j=1,2}, \\
M^{(1,1)}_{\theta}(\xi) &:= \begin{dcases} \frac{\sqrt{-1}\sgn(\xi) C_{2}(\theta)}{\sqrt{C_{1}(\theta)}} |2\pi \xi|^{\frac{5 - \theta}{2}} \quad &2 < \theta < 3, \\
\sqrt{-1}\sgn(\xi) \sqrt{C_{1}(3)} |2\pi \xi| \log{|2\pi \xi|^{-1}} + \frac{\sqrt{-1}\sgn(\xi) C_{2}(3)}{\sqrt{C_{1}(3)}} |\pi \xi| \quad &\theta = 3, \\
\frac{\sqrt{-1}\sgn(\xi) C_{2}(\theta)}{\sqrt{C_{1}(\theta)}} |2\pi \xi|^{\theta - 2} \quad &3 < \theta < 4, \\ 
\frac{\sqrt{-1}\sgn(\xi) C_{2}(\theta)}{\sqrt{C_{1}(\theta)}} (2\pi \xi)^{2} - \frac{3 \gamma}{2} (2\pi \xi)^{2} \quad &\theta = 4, \\
- \frac{3 \gamma}{2} (2\pi \xi)^{2} \quad &\theta > 4, \end{dcases} \\
M^{(1,2)}_{\theta}(\xi) &= M^{(2,1)}_{\theta}(\xi) := \begin{dcases} 0 \quad &2 < \theta < 4, \\ 
- \frac{3 \gamma}{2} (2 \pi \xi)^{2} \quad &\theta \ge 4, \end{dcases} \quad M^{(2,2)}_{\theta}(\xi) := M^{(1,1)}_{\theta}(\xi)^{*}.
\end{align}
In other words, $(\bar{F}^{+}, \bar{F}^{-})$ is the solution of the following system of linear differential equations: 
\begin{align}\label{eq:flucnormal}
\partial_{t} \bar{F}^{\iota}(y,t) &= \begin{dcases} \sgn(\iota) \frac{C_{2}(\theta)}{\sqrt{C_{1}(\theta)}} \mathcal{D}_{6-\theta} \bar{F}^{\iota} & 2 < \theta < 3, \\
\sgn(\iota) \Big( \sqrt{C_{1}(3)} \mathcal{D}_{3,l} \bar{F}^{\iota} + \frac{C_{2}(3)}{\sqrt{C_{1}(3)}} \partial_{y} \bar{F}^{\iota} \Big) & \theta = 3, \\
\sgn(\iota) \frac{C_{2}(\theta)}{\sqrt{C_{1}(\theta)}} \mathcal{D}_{2\theta - 3} \bar{F}^{\iota} & 3 < \theta < 4 \\
\sgn(\iota) \frac{C_{2}(4)}{\sqrt{C_{1}(4)}} \mathcal{H}_{\R}(\Delta \bar{F}^{\iota}) + \frac{3 \gamma}{2}\Delta (  \bar{F}^{+} + \bar{F}^{-} ) & \theta = 4, \\
\frac{3 \gamma}{2} \Delta (  \bar{F}^{+} + \bar{F}^{-} ) & \theta > 4,   \end{dcases} \\
\bar{F}^{\pm}(y,0) &= \bar{p}_{0}(y) \pm \bar{l}_{0}(y) 
\end{align}
for $\iota = \pm$ where $\mathcal{D}_{3,l}$ is defined via its Fourier transform
\begin{align}
\widetilde{(\mathcal{D}_{3,l}f)}(\xi) := \sqrt{-1}\sgn(\xi) |2\pi \xi| \log{|2\pi \xi|^{-1}} \tilde{f}(\xi).
\end{align}
Recall that $m_{\theta}(\e), n_{\theta}(\e) $ is defined in \eqref{ass:timescalingflucnormal}.

\begin{theorem}\label{thm:flucnormal}
Assume that $\theta > 2$, $\gamma > 0$ and \eqref{ass:phononic}. Then for any $t \ge 0$ and $J \in C^{\infty}_{0}(\R)$ we have
\begin{align}
\lim_{\e \to 0} \e \sum_{x \in \Z} \mathbb{E}_{\e}[f^{\pm}_{x}(\frac{t}{n_{\theta}(\e)})] \Big(S^{\pm}_{\theta}(\frac{t}{m_{\theta}(\e)})J\Big)(\e x) = \int_{\R} dy ~ \bar{F}^{\pm}(y,t) J(y).
\end{align} 
\end{theorem}

\begin{remark}
We see that if $2 < \theta < 4$ then the right-hand side of \eqref{eq:flucnormal} does not depend on the strength of the noise $\gamma$, and thus the fluctuations of normal modes are determined by the interaction potential $\a$. On the other hand, if $\theta > 4$ then the fluctuations do not depend on $\a$. The threshold is $\theta = 4$ and in this case the fluctuations depend on both $\a$ and $\gamma$.
\end{remark}

In addition, when $2 < \theta < 4$ we have the following stronger result, which means \textit{the law of large numbers} for the empirical measure of the recentered dynamics. 
\begin{theorem}\label{thm:flucnormalLLN}
Assume that $2 <\theta < 4$, $\gamma \ge 0$ and \eqref{ass:phononic}. Then for any $t \ge 0$ and $J \in C^{\infty}_{0}(\R)$ we have
\begin{align}
\lim_{\e \to 0} \mathbb{E}_{\e}\Bigg[ \Big| \e \sum_{x \in \Z} f^{\pm}_{x}(\frac{t}{n_{\theta}(\e)})\Big(S^{\pm}_{\theta}(\frac{t}{m_{\theta}(\e)})J\Big)(\e x) - \int_{\R} dy ~ \bar{F}^{\pm}(y,t) J(y) \Big|  \Bigg] = 0.
\end{align}
\end{theorem}

\section{proof of theorem \ref{thm:superballistic1}}\label{sec:proofofsuperballistic}

In this section we prove Theorem \ref{thm:superballistic1} by following the strategy presented in \cite[Section 4]{KO2}. Since if $1 < \theta \le 3$ then the asymptotic behavior of $\hat{\a}(k), k \to 0$ is different from that of exponential decay models, we need some modification. The asymptotic behavior of $\hat{\a}(k), k \to 0$ is computed in Appendix \ref{app:asympofa}.

\subsection{Mean dynamics}

First we define the scaled mean dynamics $\{ \bar{p}_{\e,x}(t) , \bar{l}_{\e,x}(t) ; x \in \Z , t \ge 0 \}$ as
\begin{align}
\bar{p}_{\e,x}(t) := \mathbb{E}_{\e} \Big[p_{x}(\frac{t}{j_{\theta}(\e)}) \Big],  \quad \quad \bar{l}_{\e,x}(t) := \mathbb{E}_{\e} \Big[l_{x}(\frac{t}{j_{\theta}(\e)}) \Big].
\end{align}
We also introduce the Fourier transform of the mean dynamics $\{ \hat{\bar{p}}_{\e}(k,t) , \hat{\bar{l}}_{\e}(k,t) ; k \in \frac{\T}{\e} , t \ge 0 \}$, that is,
\begin{align}
\hat{\bar{p}}_{\e}(k,t) := \e \mathbb{E}_{\e} \Big[\hat{p}(\e k,\frac{t}{j_{\theta}(\e)}) \Big] \quad \quad \hat{\bar{l}}_{\e}(k,t) := \e \mathbb{E}_{\e} \Big[\hat{l}(\e k,\frac{t}{j_{\theta}(\e)}) \Big].
\end{align}

The following Proposition states the scaling limit for the mean dynamics.
\begin{proposition}\label{prop:meandynamics} 
Suppose that $\theta > 1$, $\gamma \ge 0$ and \eqref{ass:phononic}. For any $t \ge 0$, we have
\begin{align}\label{lim:meanconvonT}
\lim_{\e \to 0} \int_{\frac{\T}{\e}} dk ~ \Big| \hat{\bar{p}}_{\e}(k,t) - \tilde{\bar{p}}(k,t) \Big|^2 + \Big| \hat{\bar{l}}_{\e}(k,t) -  \tilde{\bar{l}}(k,t) \Big|^{2} = 0, 
\end{align}
and consequently we have
\begin{align}\label{meanconvonZ}
\lim_{\e \to 0} \e \sum_{x \in \Z} \Big( \Big| \bar{p}_{\e,x}(t) - \bar{p}(\e x,t) \Big|^2 + \Big| \bar{l}_{\e,x}(t) - \bar{l}(\e x,t) \Big|^{2} \Big) = 0.
\end{align}
\end{proposition}
In the rest of this subsection we show that Proposition \ref{prop:meandynamics} implies Theorem \ref{thm:superballistic1}. In Section \ref{sec:proofofmeandynamics} we prove Proposition \ref{prop:meandynamics}. 

\subsubsection{Proof of Theorem \ref{thm:superballistic1}}
First we define $ \hat{P}_{\e}(k,t) , \hat{L}_{\e}(k,t) , k \in \T, t \ge 0$ as
\begin{align}
\hat{P}_{\e}(k,t) &:= \e \hat{p}(\e k,\frac{t}{j_{\theta}(\e)}) - \hat{\bar{p}}_{\e}(k,t), \\
\hat{L}_{\e}(k,t) &:= \e \hat{l}(\e k,\frac{t}{j_{\theta}(\e)}) - \hat{\bar{l}}_{\e}(k,t).
\end{align}
From \eqref{eq:energyconservation}, for any $t \ge 0$ we have 
\begin{align}
&\int_{\frac{\T}{\e}} dk ~ \mathbb{E}_{\e}\Big[ |\e \hat{p}(\e k,0)|^{2} + |\e \hat{l}(\e k,0)|^{2} \Big] \\ 
&= \int_{\frac{\T}{\e}} dk ~ \mathbb{E}_{\e}\Big[ |\e \hat{p}(\e k,\frac{t}{j_{\theta}(\e)})|^{2} + |\e \hat{l}(\e k,\frac{t}{j_{\theta}(\e)})|^{2}\Big] \\
&= \int_{\frac{\T}{\e}} dk ~ \mathbb{E}_{\e}\Big[ |\hat{P}_{\e}(k,t)|^{2} + |\hat{L}_{\e}(k,t)|^{2} \Big] + \int_{\frac{\T}{\e}} dk ~ |\hat{\bar{p}}_{\e}(k,t)|^{2} + |\hat{\bar{l}}_{\e}(k,t)|^{2} .
\end{align}
By using \eqref{ass:phononic} and Proposition \ref{prop:meandynamics}, we get
\begin{align}
&\int_{\R} dk ~ |\tilde{\bar{p}}_{0}(k)|^{2} + |\tilde{\bar{l}}_{0}(k)|^{2} = \varlimsup_{\e \to 0} \int_{\frac{\T}{\e}} dk ~ \mathbb{E}_{\e}\Big[ |\hat{P}_{\e}(k,t)|^{2} + |\hat{L}_{\e}(k,t)|^{2} \Big] + \int_{\R} dk ~ |\tilde{\bar{p}}(k,t)|^{2} + |\tilde{\bar{l}}(k,t)|^{2}.
\end{align}
Hence from \eqref{eq:energyconservation} we obtain
\begin{align}\label{lim:convoffrac}
\varlimsup_{\e \to 0} \int_{\frac{\T}{\e}} dk ~ \mathbb{E}_{\e}\Big[|\hat{P}_{\e}(k,t)|^{2} + |\hat{L}_{\e}(k,t)|^{2} \Big] = 0.
\end{align}
Combining \eqref{lim:meanconvonT} with \eqref{lim:convoffrac}, we have 
\begin{align}
\lim_{\e \to 0} \int_{\frac{\T}{\e}} dk ~ \mathbb{E}_{\e}\Big[ \Big| \e \hat{p}(\e k,\frac{t}{j_{\theta}(\e)}) - \tilde{\bar{p}}(k,t) \Big|^2 + \Big| \e \hat{l}(\e k,\frac{t}{j_{\theta}(\e)}) -  \tilde{\bar{l}}(k,t) \Big|^{2} \Big] = 0, 
\end{align}
and thus we established Theorem \ref{thm:superballistic1}.

\subsection{Proof of Proposition \ref{prop:meandynamics}}\label{sec:proofofmeandynamics}

From \eqref{def:dynamicsofpsi}, we have the time evolution law of $\{ \hat{\bar{p}}_{\e}(k,t) , \hat{\bar{l}}_{\e}(k,t) ; k \in \frac{\T}{\e} , t \ge 0 \}$:
\begin{align}\label{evoofmeanont}
\frac{d}{dt} \begin{pmatrix} \hat{\bar{p}}_{\e}(k,t) \\ \hat{\bar{l}}_{\e}(k,t) \end{pmatrix} = A_{\e}(k) \begin{pmatrix} \hat{\bar{p}}_{\e}(k,t) \\ \hat{\bar{l}}_{\e}(k,t) \end{pmatrix} ,
\end{align}
where
\begin{align}
A_{\e}(k) := \frac{1}{j_{\theta}(\e)} \begin{pmatrix} -2 \gamma R(\e k) & \sqrt{-1} \operatorname{sgn}(k) \omega(\e k) \\ \sqrt{-1} \operatorname{sgn}(k) \omega(\e k) & 0 \end{pmatrix} .
\end{align}
Then we decompose $A_{\e}$ into two parts as follows:
\begin{align}
A_{\e}(k) &= A_{\theta}(k) + B_{\e,\theta}(k), \quad A_{\theta} = (A^{(i,j)}_{\theta})_{i,j=1,2}, ~ B_{\e,\theta} = (B^{(i,j)}_{\e,\theta})_{i,j=1,2}
\end{align}
where
\begin{align}
A^{(1,1)}_{\theta}(k) & = A^{(2,2)}_{\theta}(k) \equiv 0, \\
A^{(1,2)}_{\theta}(k) &= A^{(2,1)}_{\theta}(k) := \begin{dcases} \sqrt{C_{1}(\theta)} \sqrt{-1} \sgn(k) |2\pi k|^{\frac{\theta - 1}{2}} \quad &1 < \theta \le 3, \\ \sqrt{C_{1}(\theta)} \sqrt{-1} \sgn(k) |2\pi k| \quad &\theta > 3,\end{dcases} \\
B^{(i,j)}_{\e,\theta}(k) &:= A^{(i,j)}_{\e}(k) - A^{(i,j)}_{\theta}(k) \quad i,j=1,2.
\end{align}

To prove Proposition \ref{prop:meandynamics}, we need the following Lemma \ref{lem:expnorm} and Lemma \ref{lem:finiteinterval}. 
\begin{lemma}\label{lem:expnorm}
\begin{align}
C_{*} &:= \sup_{0 < \e < 1} \sup_{(k,t) \in \R \times \R} \| \exp{\Big(A_{\e}(k) t \Big)} \| < \infty , \\
D_{*} &:= \sup_{(k,t) \in \R \times \R} \| \exp{\Big(A(k) t \Big)} \| < \infty
\end{align}
where $||\cdot||$ is the matrix norm defined as
\begin{align}
\| A \| := \sup_{x \in \R^{2}} \frac{| Ax |}{| x |} , \quad A \in M_{2}(\R).
\end{align}
In addition, for any $K > 0$ we have
\begin{align}
\sup_{|k| \le K} \| B_{\e,\theta} (k) \| &\le c_{K,\theta} b_{\theta}(\e),
\end{align}
where 
\begin{align}\label{scale:matrem}
b_{\theta}(\e) := \begin{dcases} \e \quad & 1 < \theta < 2, \\
\e \log{\e^{-1}} \quad & \theta = 2, \\
\e^{3 - \theta} \quad & 2 < \theta < 3 \\
(\log{\e^{-1}})^{-1} \quad & \theta = 3, \\
\e^{\theta - 3} \quad & 3 < \theta < 4, \\
\e \quad & \theta = 4, \\
\e^{\theta - 3} \quad & 4 < \theta < 5, \\
\e^{2} \log{\e^{-1}} \quad & \theta = 5, \\
\e^{2}  \quad &  \theta > 5. 
\end{dcases} 
\end{align}
and $c_{K,\theta}$ is a positive constant which depends on $K > 0, ~ \theta > 1$.
\end{lemma}
\begin{lemma}\label{lem:finiteinterval}
Suppose that $(\hat{\bar{p}}_{\e}^{(0)}(k,t), \hat{\bar{l}}_{\e}^{(0)}(k,t) )$ is the solution of the following equation:
\begin{align}
\frac{d}{dt} \begin{pmatrix} \hat{\bar{p}}_{\e}^{(0)}(k,t) \\ \hat{\bar{l}}_{\e}^{(0)}(k,t) \end{pmatrix} &= A(k) \begin{pmatrix} \hat{\bar{p}}_{\e}^{(0)}(k,t) \\ \hat{\bar{l}}_{\e}^{(0)}(k,t) \end{pmatrix}, \\
\begin{pmatrix} \hat{\bar{p}}_{\e}^{(0)}(k,0) \\ \hat{\bar{l}}_{\e}^{(0)}(k,0) \end{pmatrix}  &= \begin{pmatrix} \hat{\bar{p}}_{\e}(k,0) \\ \hat{\bar{l}}_{\e}(k,0) \end{pmatrix}.
\end{align}
Then for any $K > 0$, we have
\begin{align}
\lim_{\e \to 0} \int_{|k| \le K} dk ~ \Big| \hat{\bar{p}}_{\e}(k,t) - \hat{\bar{p}}_{\e}^{(0)}(k,t) \Big|^2 + \Big| \hat{\bar{l}}_{\e}(k,t) -  \hat{\bar{l}}_{\e}^{(0)}(k,t) \Big|^{2} = 0.
\end{align}
\end{lemma}

From now on we prove Lemma \ref{lem:expnorm} and Lemma \ref{lem:finiteinterval}, and then we verify Proposition  \ref{prop:meandynamics} by using these lemmas.

\subsubsection{Proof of Lemma \ref{lem:expnorm}}

\begin{proof}

Since the eigenvalues of $A(k)$ are imaginary for every $k \in \T$, we have $D_{*} < \infty$. The eigenvalues of $A_{\e}(k)$ are 
\begin{align}
- \frac{1}{j_{\theta}(\e)} \Big( \gamma R(\e k) \pm \sqrt{\gamma^{2}R^{2}(\e k) - \hat{\a}(\e k)} \Big).
\end{align} 
Since $R(\cdot)$ and $\hat{\a}(\cdot)$ are positive, the real parts of the eigenvalues of $A_{\e}(k)$ are negative. Therefore we have $C_{*} < \infty$.

Next we consider the order of $\sup_{|k| \le K} \| B_{\e,\theta}(k) \|, ~ K > 0$. From Lemma \ref{lem:asympofa} we obtain
\begin{align}
| B^{(1,2)}_{\e}(k) | &= \begin{dcases} \Big| \frac{C_{1}(\theta) |2\pi k|^{\theta - 1} - \hat{\a}(\e k) j_{\theta}(\e)^{-2}}{ \sqrt{-1} \operatorname{sgn}(k) \omega(\e k) j_{\theta}(\e)^{-1} + \sqrt{C_{1}(\theta)} \sqrt{-1} \sgn(k) |2\pi k|^{\frac{\theta - 1}{2}} } \Big| & 1 < \theta \le 3, \\ \Big| \frac{C_{1}(\theta) |2\pi k|^{2} - \hat{\a}(\e k) \e^{-2} }{ \sqrt{-1} \operatorname{sgn}(k) \omega(\e k) \e^{-1} + \sqrt{C_{1}(\theta)} \sqrt{-1} (2\pi k) } \Big| & \theta > 3, \end{dcases} \\
&\lesssim \begin{dcases} \e \quad & 1 < \theta < 2, \\
\e \log{\e^{-1}} \quad & \theta = 2, \\
\e^{3 - \theta} \quad & 2 < \theta < 3 \\
(\log{\e^{-1}})^{-1} \quad & \theta = 3, \\
\e^{\theta - 3} \quad & 3 < \theta < 4, \\
\e \quad & \theta = 4, \\
\e^{\theta - 3} \quad & 4 < \theta < 5, \\
\e^{2} \log{\e^{-1}} \quad & \theta = 5, \\
\e^{2}  \quad &  \theta > 5, 
\end{dcases}
\end{align}
and
\begin{align} 
|B^{(1,1)}_{\e}(k)| &\lesssim \e^{2}
\end{align}
on $|k| \le K$. Hence we complete the proof of this lemma.

\end{proof}

\subsubsection{Proof of Lemma \ref{lem:finiteinterval}}
\begin{proof}

Fix a positive constant $K > 0$ and $0 < \e < 1$ such that $K < \frac{1}{2 \e}$. From \eqref{evoofmeanont} and the decomposition of $A_{\e}(k)$, we have
\begin{align}
\frac{d}{dt} \begin{pmatrix} \hat{\bar{p}}_{\e}(k,t)  - \hat{\bar{p}}_{\e}^{(0)}(k,t)  \\ \hat{\bar{l}}_{\e}(k,t)  - \hat{\bar{l}}_{\e}^{(0)}(k,t)  \end{pmatrix} = A(k) \begin{pmatrix} \hat{\bar{p}}_{\e}(k,t) - \hat{\bar{p}}_{\e}^{(0)}(k,t) \\ \hat{\bar{l}}_{\e}(k,t) - \hat{\bar{l}}_{\e}^{(0)}(k,t) \end{pmatrix}  + B_{\e,\theta}(k) \begin{pmatrix} \hat{\bar{p}}_{\e}(k,t) \\ \hat{\bar{l}}_{\e}(k,t) \end{pmatrix}.
\end{align}
By using Duhamel's formula, we have
\begin{align}
\begin{pmatrix} \hat{\bar{p}}_{\e}(k,t) - \hat{\bar{p}}_{\e}^{(0)}(k,t) \\ \hat{\bar{l}}_{\e}(k,t) - \hat{\bar{l}}_{\e}^{(0)}(k,t) \end{pmatrix}  = \int_{0}^{t} ds ~ \exp{ \Big( (t-s)A(k) \Big) } B_{\e,\theta}(k) \begin{pmatrix} \hat{\bar{p}}_{\e}(k,s) \\ \hat{\bar{l}}_{\e}(k,s) \end{pmatrix}.
\end{align}
From \eqref{eq:energyconservation} and Lemma \ref{lem:expnorm}, we have
\begin{align}
&\int_{|k| \le K} dk ~ \Big| \hat{\bar{p}}_{\e}(k,t) - \hat{\bar{p}}_{\e}^{(0)}(k,t) \Big|^2 + \Big| \hat{\bar{l}}_{\e}(k,t) -  \hat{\bar{l}}_{\e}^{(0)}(k,t) \Big|^{2}  \\
&\le t \int_{|k| \le K} dk \int_{0}^{t} ds ~ \Big| \exp{ \Big( (t-s)A(k) \Big) } B_{\e,\theta}(k) \begin{pmatrix} \hat{\bar{p}}_{\e}(k,s) \\ \hat{\bar{l}}_{\e}(k,s) \end{pmatrix} \Big|^{2} \\
&\lesssim t (b_{\theta}(\e))^{2} \int_{0}^{t} ds \int_{\T} dk  ~ \e \mathbb{E}_{\e} \Big[ |\hat{\psi}(k,\frac{s}{j_{\theta}(\e)})|^{2} \Big] \\
&\lesssim t^{2} (b_{\theta}(\e))^{2}.
\end{align}
We established this lemma.
\end{proof}

\subsubsection{Proof of Proposition \ref{prop:meandynamics}}
\begin{proof}

By using Schwarz's inequality we have
\begin{align}
&\int_{\frac{\T}{\e}} dk ~ \Big| \hat{\bar{p}}_{\e}(k,t) - \tilde{\bar{p}}(k,t) \Big|^2 + \Big| \hat{\bar{l}}_{\e}(k,t) -  \tilde{\bar{l}}(k,t) \Big|^{2} \\
&\le 2 \int_{\frac{\T}{\e}} dk ~ \Big| \hat{\bar{p}}_{\e}(k,t) - \hat{\bar{p}}_{\e}^{(0)}(k,t) \Big|^2 + \Big| \hat{\bar{l}}_{\e}(k,t) -  \hat{\bar{l}}_{\e}^{(0)}(k,t) \Big|^{2} \\
& \quad + 2 \int_{\frac{\T}{\e}} dk ~ \Big| \hat{\bar{p}}_{\e}^{(0)}(k,t) - \tilde{\bar{p}}(k,t) \Big|^2 + \Big| \hat{\bar{l}}_{\e}^{(0)}(k,t) -  \tilde{\bar{l}}(k,t) \Big|^{2} \\
&\le 2 \int_{\frac{\T}{\e}} dk ~ \Big| \hat{\bar{p}}_{\e}(k,t) - \hat{\bar{p}}_{\e}^{(0)}(k,t) \Big|^2 + \Big| \hat{\bar{l}}_{\e}(k,t) -  \hat{\bar{l}}_{\e}^{(0)}(k,t) \Big|^{2}  \\
& \quad + 2 \int_{\frac{\T}{\e}} dk ~ \Big\| \exp{\Big( A(k) t \Big)} \Big\|^{2} \Big| \begin{pmatrix} \hat{\bar{p}}_{\e}(k,0) \\ \hat{\bar{l}}_{\e}(k,0) \end{pmatrix} - \begin{pmatrix} \tilde{\bar{p}}_{0}(k) \\ \tilde{\bar{l}}_{0}(k) \end{pmatrix} \Big|^{2} \label{ineq:meandynamics}.
\end{align}
From \eqref{ass:phononic} and Lemma \ref{lem:expnorm}, we see that the second term of \eqref{ineq:meandynamics} vanishes. Now we estimate the first term. From Lemma \ref{lem:finiteinterval}, it suffices to show that
\begin{align}\label{sufflem51}
\varlimsup_{K \to \infty} \varlimsup_{\e \to 0} \int_{K < |k| < \frac{1}{2 \e}} dk ~ \Big| \hat{\bar{p}}_{\e}(k,t) - \hat{\bar{p}}_{\e}^{(0)}(k,t) \Big|^2 + \Big| \hat{\bar{l}}_{\e}(k,t) -  \hat{\bar{l}}_{\e}^{(0)}(k,t) \Big|^{2} = 0. 
\end{align} 
From \eqref{evoofmeanont}, we can write 
\begin{align}
\begin{pmatrix} \hat{\bar{p}}_{\e}(k,t) \\ \hat{\bar{l}}_{\e}(k,t) \end{pmatrix} = \exp{\Big( A_{\e}(k) t \Big)} \begin{pmatrix} \hat{\bar{p}}_{\e}(k,0) \\ \hat{\bar{l}}_{\e}(k,0) \end{pmatrix}.
\end{align} 
In addition, $(\tilde{\bar{p}}(\xi,t), \tilde{\bar{l}}(\xi,t)), \xi \in \R, t \ge 0$ satisfies
\begin{align}
\partial_{t} \begin{pmatrix} \tilde{\bar{p}}(\xi,t) \\ \tilde{\bar{p}}(\xi,t) \end{pmatrix} = A(\xi) \begin{pmatrix} \tilde{\bar{p}}(\xi,t) \\ \tilde{\bar{l}}(\xi,t)\end{pmatrix},
\end{align}
and hence we obtain
\begin{align}
\begin{pmatrix} \tilde{\bar{p}}(\xi,t) \\ \tilde{\bar{p}}(\xi,t) \end{pmatrix} = \exp{\Big( A(\xi) t \Big)} \begin{pmatrix} \tilde{\bar{p}}_0(\xi) \\ \tilde{\bar{l}}_0(\xi)\end{pmatrix}.
\end{align}
By using \eqref{ass:phononic}, Lemma \ref{lem:expnorm} and Schwarz's inequality we obtain
\begin{align}
& \varlimsup_{\e \to 0}  \int_{K < |k| < \frac{1}{2 \e}} dk ~ \Big| \hat{\bar{p}}_{\e}(k,t) - \hat{\bar{p}}_{\e}^{(0)}(k,t) \Big|^2 + \Big| \hat{\bar{l}}_{\e}(k,t) -  \hat{\bar{l}}_{\e}^{(0)}(k,t) \Big|^{2} \\
& =  \varlimsup_{\e \to 0}  \int_{K < |k| < \frac{1}{2 \e}} dk ~ \Big|  \exp{\Big( A_{\e}(k) t \Big)}  \begin{pmatrix} \hat{\bar{p}}_{\e}(k,0) \\ \hat{\bar{l}}_{\e}(k,0) \end{pmatrix} - \exp{\Big( A(k) t \Big)} \begin{pmatrix} \hat{\bar{p}}_{\e}(k,0) \\ \hat{\bar{l}}_{\e}(k,0) \end{pmatrix}  \Big|^{2} \\
&\le 2 \varlimsup_{\e \to 0}  \int_{K < |k| < \frac{1}{2 \e}} dk ~ \Big| \exp{\Big( A_{\e}(k) t \Big)} \begin{pmatrix} \hat{\bar{p}}_{\e}(k,0) \\ \hat{\bar{l}}_{\e}(k,0) \end{pmatrix} \Big|^{2} +  \Big| \exp{\Big( A(k) t \Big)} \begin{pmatrix} \hat{\bar{p}}_{\e}(k,0) \\ \hat{\bar{l}}_{\e}(k,0) \end{pmatrix} \Big|^{2} \\
&\le 2(C_{*} + D_{*})  \varlimsup_{\e \to 0} \int_{K < |k| < \frac{1}{2 \e}} dk ~ |\hat{\bar{p}}_{\e}(k,0)|^{2} + |\hat{\bar{l}}_{\e}(k,0)|^{2} \\
&\le 2(C_{*} + D_{*}) \int_{|\xi| > K} d\xi ~ |\tilde{\bar{p}}_{0}(\xi)|^{2} + |\tilde{\bar{l}}_{0}(\xi)|^{2},
\end{align}
and thus we have \eqref{sufflem51}.
\end{proof}

\section{Proof of Theorem \ref{thm:superballisticenergy}}

First we observe that from Theorem \ref{thm:superballistic1} it is sufficient to show that
\begin{align}
&\lim_{\e \to 0} \e \sum_{x \in \Z} \mathbb{E}_{\e}\Big[e_{x}(\frac{t}{j_{\theta}(\e)}) - \frac{1}{2} p_{x}^2(\frac{t}{j_{\theta}(\e)}) \Big] J(\e x) \\
&= \begin{dcases} \frac{1}{4} \int_{\R} dy ~ L(y,t) J(y) & 1 < \theta < 3, \\
\frac{1}{2} \int_{\R} dy ~ \bar{l}^2(y,t) J(y) & \theta \ge 3, \end{dcases} \\
&= \begin{dcases} \int_{\R^2} d\xi dk ~ \sgn(k)\sgn(-k-\xi) \frac{|\xi|^{\theta - 1} - |k|^{\theta - 1} - |k + \xi|^{\theta - 1}}{4 |k|^{\frac{\theta - 1}{2}}|k+\xi|^{\frac{\theta - 1}{2}}}  \tilde{\bar{l}}(k,t) \tilde{\bar{l}}(-k - \xi,t) \tilde{J}(\xi) & 1 < \theta < 3, \\ \frac{1}{2} \int_{\R^2} d\xi dk ~ \tilde{\bar{l}}(k,t) \tilde{\bar{l}}(-k - \xi,t) \tilde{J}(\xi) & \theta \ge 3, \end{dcases} 
\end{align}
where $L(y,t)$ is defined in \eqref{def:potentialenergy}. To simplify the notation, we may omit the variable $t \ge 0$. By using the Poisson summation formula,  we obtain
\begin{align}
&\e \sum_{x} \mathbb{E}_{\e}\Big[e_{x}(\frac{t}{j_{\theta}(\e)}) - \frac{1}{2} p_{x}^2(\frac{t}{j_{\theta}(\e)}) \Big] J(\e x)  \\
&= \frac{\e}{4} \sum_{x \in \Z} \int_{\T^2}  dk dk^{'} ~ e^{2 \pi \sqrt{-1} (k+k^{'}) x} \sgn(k)\sgn(k^{'}) F(k,k^{'}) \mathbb{E}_{\e}\Big[ \hat{l}(k) \hat{l}(k^{'}) \Big]  \int_{\R} d\xi ~ e^{2 \pi \sqrt{-1} \xi \e x} \tilde{J}(\xi) \\
&= \frac{1}{4} \int_{\R \times \T} d\xi dk ~ \sgn(k)\sgn(-k - \e \xi) F(k,-k - \e \xi) \e \mathbb{E}_{\e}\Big[ \hat{l}(k) \hat{l}(-k - \e \xi) \Big] \tilde{J}(\xi) \\
&= \frac{1}{4} \int_{\R \times \frac{\T}{\e}} d\xi dk ~ \sgn(k)\sgn(-k - \xi) F(\e k,- \e k - \e \xi) \e^2 \mathbb{E}_{\e}\Big[ \hat{l}(\e k) \hat{l}(-\e k - \e \xi) \Big] \tilde{J}(\xi).
\end{align}
From Theorem \ref{thm:superballistic1} and the boundedness of $F(k,k^{'})$, we have
\begin{align}
&\lim_{\e \to 0} \int_{\R \times \frac{\T}{\e}} d\xi dk ~ \sgn(k)\sgn(-k - \xi) F(\e k,- \e k - \e \xi) \e^2 \mathbb{E}_{\e}\Big[ \hat{l}(\e k,\frac{t}{j_{\theta}(\e)}) \hat{l}(-\e k - \e \xi,\frac{t}{j_{\theta}(\e)}) \Big] \tilde{J}(\xi)  \\
&= \lim_{\e \to 0} \int_{\R \times \frac{\T}{\e}} d\xi dk ~ \sgn(k)\sgn(-k - \xi) F(\e k,- \e k - \e \xi) \tilde{\bar{l}}(k,t) \tilde{\bar{l}}(-k - \xi,t)  \tilde{J}(\xi).
\end{align}
Since $1_{\{ k  \in \frac{\T}{\e} \}} F(\e k,- \e k - \e \xi)$ is uniformly bounded and 
\begin{align}
&\lim_{\e \to 0} 1_{\{ k  \in \frac{\T}{\e} \}} \sgn(k)\sgn(-k - \xi) F(\e k,- \e k - \e \xi) \\ 
&= \begin{dcases} \sgn(k)\sgn(-k-\xi) \frac{|\xi|^{\theta - 1} - |k|^{\theta - 1} - |k + \xi|^{\theta - 1}}{|k|^{\frac{\theta - 1}{2}}|k+\xi|^{\frac{\theta - 1}{2}}} & 1 < \theta < 3, \\
2 & \theta \ge 3,  \end{dcases}
\end{align}
almost every $(\xi,k) \in \R^2$, we obtain 
\begin{align}
&\lim_{\e \to 0} \int_{\R \times \frac{\T}{\e}} d\xi dk ~ \sgn(k)\sgn(-k - \xi) F(\e k,- \e k - \e \xi) \tilde{\bar{l}}(k,t) \tilde{\bar{l}}(-k - \xi,t)  \tilde{J}(\xi) \\
&= \begin{dcases} \int_{\R^2} d\xi dk ~ \sgn(k)\sgn(-k-\xi) \frac{|\xi|^{\theta - 1} - |k|^{\theta - 1} - |k + \xi|^{\theta - 1}}{|k|^{\frac{\theta - 1}{2}}|k+\xi|^{\frac{\theta - 1}{2}}}  \tilde{\bar{l}}(k,t) \tilde{\bar{l}}(-k - \xi,t) \tilde{J}(\xi) & 1 < \theta < 3, \\ 2 \int_{\R^2} d\xi dk ~ \tilde{\bar{l}}(k,t) \tilde{\bar{l}}(-k - \xi,t) \tilde{J}(\xi) & \theta \ge 3 \end{dcases}
\end{align}
and thus we get \eqref{lim:superballisticenergy}.

\section{Proof of Theorem \ref{thm:flucnormal} and \ref{thm:flucnormalLLN}}

In this section we show Theorem \ref{thm:flucnormal} and \ref{thm:flucnormalLLN} by using the strategy which is similar to that used in Section \ref{sec:proofofsuperballistic}. First we observe that 
\begin{align}
&\lim_{\e \to 0} \e \sum_{x \in \Z} \mathbb{E}_{\e}[f^{\pm}_{x}(\frac{t}{n_{\theta}(\e)})] \Big(S^{\pm}_{\theta}(\frac{t}{m_{\theta}(\e)})J\Big)(\e x) = \lim_{\e \to 0} \int_{\frac{\T}{\e}} dk ~ \hat{\bar{F}}^{\pm}_{\e}(k,t) \tilde{J}(-k),
\end{align}
where $\hat{\bar{F}}^{\pm}_{\e}(k,t), k \in \frac{\T}{\e}, t \ge 0$ is defined as
\begin{align}
\hat{\bar{F}}^{\pm}_{\e}(k,t) &:= \e \mathbb{E}_{\e}\Big[ \hat{F}^{\pm}_{\e}(k,t) \Big] , \\
\hat{F}^{\pm}_{\e}(k,t) &:= \begin{dcases} \exp\Big({\frac{\mp \sqrt{C_{1}(\theta)} \sqrt{-1} \sgn(k) |2\pi k|^{\frac{\theta - 1}{2}}t}{m_{\theta}(\e)}}\Big) \hat{f}^{\pm}_{\e}(k,t) & 2 < \theta \le 3, \\ \exp\Big({\frac{\mp \sqrt{C_{1}(\theta)} \sqrt{-1} (2\pi k) t}{m_{\theta}(\e)}}\Big) \hat{f}^{\pm}_{\e}(k,t) & \theta > 3, \end{dcases}
\end{align}
and 
\begin{align}
\hat{f}^{\pm}_{\e}(k,t) := \e \hat{f}^{\pm}(\e k, \frac{t}{n_{\theta}(\e)}).
\end{align}
The main subject of this section is to show the following Proposition, which implies Theorem \ref{thm:flucnormal}.

\begin{proposition}\label{prop:toshowfluc}
Suppose that $\theta > 2$, $\gamma > 0$ and \eqref{ass:phononic}. For any $t \ge 0$, we have
\begin{align}
\varlimsup_{\e \to 0} \int_{\frac{\T}{\e}} dk ~ \Big| \begin{pmatrix}\hat{\bar{F}}^{+}_{\e}(k,t) \\ \hat{\bar{F}}^{-}_{\e}(k,t) \end{pmatrix} - \begin{pmatrix}\tilde{\bar{F}}^{+}(k,t) \\ \tilde{\bar{F}}^{-}(k,t) \end{pmatrix} \Big|^{2} = 0.
\end{align}

\end{proposition}

We prove Proposition \ref{prop:toshowfluc} in Section \ref{subsec:proofofflucmean}. In Section \ref{subsec:proofofflucnormalLLN} we show Theorem \ref{thm:flucnormalLLN} by using Proposition \ref{prop:toshowfluc}. In the rest of this subsection, we consider the time evolution law of $\{ \hat{\bar{F}}^{\pm}_{\e}(k,t) ; k \in \frac{\T}{\e}, t \ge 0 \}$ and prepare two lemmas to show Proposition \ref{prop:toshowfluc}.
 
From \eqref{def:dynamicsofpsi} we have
\begin{align}
\frac{d}{dt} \begin{pmatrix}\hat{\bar{f}}^{+}_{\e} \\ \hat{\bar{f}}^{-}_{\e} \end{pmatrix} = \frac{1}{n_{\theta}(\e)} \begin{pmatrix} \sqrt{-1}\sgn(k)\omega(\e k) - \gamma R(\e k) & - \gamma R(\e k) \\ - \gamma R(\e k) & - \sqrt{-1}\sgn(k)\omega(k) - \gamma R(\e k) \end{pmatrix} \begin{pmatrix}\hat{\bar{f}}^{+}_{\e} \\ \hat{\bar{f}}^{-}_{\e} \end{pmatrix}.
\end{align}
Thus the dynamics $\{ \hat{\bar{F}}^{\pm}_{\e}(k,t); k \in \frac{\T}{\e}, t \ge 0 \}$ is given by
\begin{align}
\frac{d}{dt} \begin{pmatrix}\hat{\bar{F}}^{+}_{\e} \\ \hat{\bar{F}}^{-}_{\e} \end{pmatrix}(k,t) &= M_{\e,\theta}(k) \begin{pmatrix}\hat{\bar{F}}^{+}_{\e} \\ \hat{\bar{F}}^{-}_{\e} \end{pmatrix}(k,t), \quad \quad M_{\e,\theta}(k) = (M^{(i,j)}_{\e,\theta}(k))_{i,j=1,2}, \\
M^{(1,1)}_{\e,\theta}(k) &:= \frac{\sqrt{-1}\sgn(k)\omega(\e k)}{n_{\theta}(\e)} - \frac{\sqrt{C_{1}(\theta)} \sqrt{-1} \operatorname{sgn}(k) |2 \pi k|^{\frac{\theta - 1}{2} \wedge 1}}{m_{\theta}(\e)}  - \frac{\gamma R(\e k)}{n_{\theta}(\e)} , \\
M^{(1,2)}_{\e,\theta}(k) &= M^{(2,1)}_{\e,\theta}(k) := - \frac{\gamma R(\e k)}{n_{\theta}(\e)}, \quad M^{(2,2)}_{\e,\theta}(k) := (M^{(1,1)}_{\e,\theta})^{*}(k).
\end{align}
Then we have the following decomposition of $M_{\e,\theta}(k)$:
\begin{align}
M_{\e,\theta}(k) &= M_{\theta}(k) + \operatorname{Rem}_{\e,\theta}(k), \quad M_{\theta} = (M^{(i,j)}_{\theta})_{i,j=1,2}, ~ \operatorname{Rem}_{\e,\theta} = (\operatorname{Rem}^{(i,j)}_{\e,\theta})_{i,j=1,2},\\
\operatorname{Rem}^{(i,j)}_{\e,\theta}(k) &:= M^{(i,j)}_{\e,\theta}(k) - M^{(i,j)}_{\theta}(k) \quad i,j = 1,2,
\end{align}
where the matrix $M_{\theta}$ is defined right before Theorem \ref{thm:flucnormal}. We can obtain the following estimates of the matrix norm of $M_{\theta}, M_{\e,\theta}$ and $\operatorname{Rem}_{\e,\theta}$.

\begin{lemma}\label{lem:boundofM}
\begin{align}
C^{'}_{*} &:= \sup_{0 < \e < 1} \sup_{(k,t) \in \R \times \R_{\ge 0}} \| \exp{\Big( M_{\e,\theta}(k)t \Big)} \| < \infty, \\
D^{'}_{*} &:= \sup_{(k,t) \in \R \times \R_{\ge 0}} \| \exp{\Big( M_{\theta}(k)t \Big)} \| < \infty.
\end{align}
In addition, for any $K > 0$, we have
\begin{align}\label{lem:orderofrem}
\sup_{|k| \le K} \| \operatorname{Rem}_{\e,\theta}(k) \| \le C_{K,\theta} r_{\theta}(\e),
\end{align}
where
\begin{align}
r_{\theta}(\e) := \begin{dcases} \e^{\theta - 2} & 2 < \theta \le \frac{5}{2} \\
\e^{3 - \theta} & \frac{5}{2} < \theta < 3, \\
\big(\log{(\e^{-1})} \big)^{-1} & \theta = 3, \\
\e^{\theta - 3} & 3 < \theta \le \frac{7}{2}, \\
\e^{4 - \theta} & \frac{7}{2} < \theta < 4, \\
\e \log{(\e^{-1})} & \theta = 4, \\
\e^{\theta - 4} & 4 < \theta < 5, \\
 \e \log{(\e^{-1})}  & \theta = 5, \\
\e & \theta > 5. \end{dcases}
\end{align}
\end{lemma}
\begin{proof}
$M_{\e}$ and $M$ have the form
\begin{align}
\begin{pmatrix} \sqrt{-1}a + b & b \\ b & - \sqrt{-1} a + b \end{pmatrix} \quad a, b \in \R,
\end{align}
and the eigenvalues of the matrix are $b \pm \sqrt{b^{2} - a^{2}}$. Since $R(k)$ is non-negative, we see that the eigenvalues of $M_{\e}$ and $M$ are non-positive and thus we obtain $C^{'}_{*}, D^{'}_{*} < \infty$. 

Next we consider the order of $\sup_{|k| \le K} \| \operatorname{Rem}_{\e,\theta}(k) \|$, $K > 0$. Since
\begin{align}
|\operatorname{Rem}^{(1,2)}_{\e,\theta}(k)| \lesssim |\operatorname{Rem}^{(1,1)}_{\e,\theta}(k)|,
\end{align} 
it is sufficient to estimate $|\operatorname{Rem}^{(1,1)}_{\e,\theta}(k)|$. From Lem \ref{lem:asympofa}, we have
\begin{align}
&\frac{\omega(\e k)}{n_{\theta}(\e)} - \frac{\sqrt{C_{1}(\theta)} |2 \pi k|^{\frac{\theta - 1}{2} \wedge 1}}{m_{\theta}(\e)}  =  \frac{\hat{\a}(\e k) n_{\theta}(\e)^{-2} - C_{1}(\theta)|2 \pi k|^{(\theta - 1) \wedge 2} m_{\theta}(\e)^{-2} }{\omega(\e k)n_{\theta}(\e)^{-1} + \sqrt{C_{1}(\theta)} |2 \pi k|^{\frac{\theta - 1}{2} \wedge 1}m_{\theta}(\e)^{-1} } \\
&=  \begin{dcases} \frac{C_{2}(\theta) |2 \pi k|^{2} }{\omega(\e k) j_{\theta}(\e)^{-1}  + \sqrt{C_{1}(\theta)} |2 \pi k|^{\frac{\theta - 1}{2}}} + O(\e^{\theta - 2} |k|^{\frac{\theta + 1}{2}}) & 2 < \theta < 3, \\ 
\frac{- C_{1}(3) |2 \pi k|^2 \log{(2\pi k)}  + C_{2}(3) |2 \pi k|^{2} }{\omega(\e k)  j_{\theta}(\e)^{-1}  + \sqrt{C_{1}(3)} |2 \pi k| } + O(\e |k|^2) & \theta = 3, \\  
\frac{C_{2}(\theta) |2 \pi k|^{\theta - 1} }{\omega(\e k)  j_{\theta}(\e)^{-1}  + \sqrt{C_{1}(\theta)} |2 \pi k| } + O(\e |k|^{\theta - 1})   & 3 < \theta < 4, \\
 \frac{C_{2}(4) |2 \pi k|^{3} }{\omega(\e k) j_{\theta}(\e)^{-1}  + \sqrt{C_{1}(4)} |2 \pi k|  } + O( \e \log{(\e^{-1})} |k|^{3}) & \theta = 4, \end{dcases}
\end{align}
and 
\begin{align}
\frac{\omega(\e k)}{n_{\theta}(\e)} - \frac{\sqrt{C_{1}(\theta)} |2 \pi k| }{m_{\theta}(\e)} = \begin{dcases} O(\e^{\theta - 4} |k|^{\theta - 2} + \e |k|^{3})  & 4 < \theta < 5, \\
O(\e \log{(\e^{-1})} |k|^{3}) & \theta = 5, \\
O(\e |k|^{3} + \e^{\theta - 4} |k|^{\theta - 2}) & 5 < \theta < 7, \\ 
O(\e |k|^{3} + \e^{3} \log{(\e^{-1})} |k|^{5}) & \theta = 7, \\
O(\e |k|^{3} + \e^{3} |k|^{5}) & \theta > 7. \end{dcases}
\end{align}
In addition, we have
\begin{align}
\frac{R(\e k)}{n_{\theta}(\e)} = \begin{dcases} O( \e^{2}n_{\theta}(\e)^{-1} |k|^2 + \e^{4} n_{\theta}(\e)^{-1} |k|^4) & 2 < \theta < 4, \\
\frac{3}{2}(2 \pi k)^{2} + O( \e^{2}  |k|^4 ) & \theta \ge 4. \end{dcases}
\end{align}
Thus we obtain
\begin{align}
\operatorname{Rem}^{(1,1)}_{\e,\theta}(k) =  \begin{dcases} O\Big(\e^{3 - \theta} |k|^{\frac{11 - 3 \theta}{2}} + \e^{\theta - 2} |k|^{\frac{\theta + 1}{2}} +  \e^{\frac{\theta - 1}{2}} k^2 \Big) & 2 < \theta < 3, \\
O\Big( \big(\log{(\e^{-1})} \big)^{-1} |k| + \e |k|^{2} +  \e(\sqrt{\log{(\e^{-1})}}) k^2 \Big) & \theta = 3, \\
O\Big(\e^{\theta - 3} |k|^{2\theta - 5} + \e |k|^{\theta - 1} +  \e^{4 - \theta} k^2 \Big) & 3 < \theta < 4, \\
O\Big( \e \log{(\e^{-1})} |k|^{3} + \e^{2}  |k|^4 \Big) & \theta = 4, \\
O\Big( \e^{\theta - 4} |k|^{\theta - 2} + \e |k|^{3} + \e^{2}  |k|^4 \Big) & 4 < \theta < 5, \\
O\Big( \e \log{(\e^{-1})} |k|^{3} + \e^{2}  |k|^4 \Big) & \theta = 5, \\
O\Big( \e |k|^{3} + \e^{\theta - 4} |k|^{\theta - 2} + \e^{2}  |k|^4 \Big) & 5 < \theta < 7, \\
O\Big( \e |k|^{3} + \e^{3} \log{(\e^{-1})} |k|^{5} + \e^{2}  |k|^4 \Big) & \theta = 7, \\
O\Big( \e |k|^{3} + \e^{3} |k|^{5} + \e^{2}  |k|^4 \Big) & \theta > 7. 
\end{dcases}
\end{align}
Hence on $|k| \le K$ we have
\begin{align}
|\operatorname{Rem}^{(1,1)}_{\e,\theta}(k)| \lesssim \begin{dcases} \e^{\theta - 2} & 2 < \theta \le \frac{5}{2} \\
\e^{3 - \theta} & \frac{5}{2} < \theta < 3, \\
\big(\log{(\e^{-1})} \big)^{-1} & \theta = 3, \\
\e^{\theta - 3} & 3 < \theta \le \frac{7}{2}, \\
\e^{4 - \theta} & \frac{7}{2} < \theta < 4, \\
\e \log{(\e^{-1})} & \theta = 4, \\
\e^{\theta - 4} & 4 < \theta < 5, \\
 \e \log{(\e^{-1})}  & \theta = 5, \\
\e & \theta > 5. \end{dcases}
\end{align}

\end{proof}

We introduce the dynamics $\{ \hat{\bar{F}}^{\pm,(0)}_{\e}(k,t) ; k \in \frac{\T}{\e}, t \ge 0 \}$, which is generated by $M(k)$ with the same initial condition as $\{ \hat{\bar{F}}^{\pm}_{\e}(k,t) ; k \in \frac{\T}{\e}, t \ge 0 \}$, that is, 
\begin{align}
\frac{d}{dt} \begin{pmatrix}\hat{\bar{F}}^{+,(0)}_{\e}(k,t) \\ \hat{\bar{F}}^{-,(0)}_{\e}(k,t) \end{pmatrix}(k,t) &= M_{\theta}(k) \begin{pmatrix}\hat{\bar{F}}^{+,(0)}_{\e}(k,t) \\ \hat{\bar{F}}^{-,(0)}_{\e}(k,t) \end{pmatrix}, \\
\begin{pmatrix}\hat{\bar{F}}^{+,(0)}_{\e}(k,0) \\ \hat{\bar{F}}^{-,(0)}_{\e}(k,0) \end{pmatrix} &= \begin{pmatrix}\hat{\bar{F}}^{+}_{\e}(k,0) \\ \hat{\bar{F}}^{-}_{\e}(k,0) \end{pmatrix}.
\end{align}

\begin{lemma}\label{lem:convtomean}
For any $K > 0$ and $T > 0$, we have
\begin{align}
\lim_{\e \to 0} \int_{|k| \le K} dk ~ \| \begin{pmatrix}\hat{\bar{F}}^{+,(0)}_{\e}(k,t) \\ \hat{\bar{F}}^{-,(0)}_{\e}(k,t) \end{pmatrix} - \begin{pmatrix}\hat{\bar{F}}^{+}_{\e}(k,t) \\ \hat{\bar{F}}^{-}_{\e}(k,t) \end{pmatrix} \|^{2} = 0.
\end{align}
\end{lemma}

\begin{proof}

Since 
\begin{align}
\frac{d}{dt} \begin{pmatrix}\hat{\bar{F}}^{+}_{\e}(k,t) - \hat{\bar{F}}^{+,(0)}_{\e}(k,t) \\ \hat{\bar{F}}^{-}_{\e}(k,t) - \hat{\bar{F}}^{-,(0)}_{\e}(k,t) \end{pmatrix} = M_{\theta}(k) \begin{pmatrix}\hat{\bar{F}}^{+}_{\e}(k,t) - \hat{\bar{F}}^{+,(0)}_{\e}(k,t) \\ \hat{\bar{F}}^{-}_{\e}(k,t) - \hat{\bar{F}}^{-,(0)}_{\e}(k,t) \end{pmatrix} + \operatorname{Rem}_{\e,\theta}(k)  \begin{pmatrix}\hat{\bar{F}}^{+}_{\e}(k,t) \\ \hat{\bar{F}}^{-}_{\e}(k,t) \end{pmatrix},
\end{align}
by using Duhamel's formula we obtain
\begin{align}
\begin{pmatrix}\hat{\bar{F}}^{+}_{\e}(k,t) - \hat{\bar{F}}^{+,(0)}_{\e}(k,t) \\ \hat{\bar{F}}^{-}_{\e}(k,t) - \hat{\bar{F}}^{-,(0)}_{\e}(k,t) \end{pmatrix} = \int_{0}^{t} ds ~ \exp{\Big( M(k)(t-s)\Big)} \operatorname{Rem}_{\e,\theta}(k) \begin{pmatrix}\hat{\bar{F}}^{+}_{\e}(k,t) \\ \hat{\bar{F}}^{-}_{\e}(k,t) \end{pmatrix}.
\end{align}
Thanks to \eqref{eq:energyconservation}, \eqref{ass:phononic} and Lemma \ref{lem:boundofM}, we have
\begin{align}
& \int_{|k| \le K} dk ~ \| \begin{pmatrix}\hat{\bar{F}}^{+,(0)}_{\e}(k,t) \\ \hat{\bar{F}}^{-,(0)}_{\e}(k,t) \end{pmatrix} - \begin{pmatrix}\hat{\bar{F}}^{+}_{\e}(k,s) \\ \hat{\bar{F}}^{-}_{\e}(k,s) \end{pmatrix} \|^{2} \\
&\le t \int_{|k| \le K} dk ~ \int_{0}^{t} ds ~ \| \exp{\Big( M_{\theta}(k)(t-s) \Big)} \operatorname{Rem}_{\e,\theta}(k) \begin{pmatrix}\hat{\bar{F}}^{+}_{\e}(k,s) \\ \hat{\bar{F}}^{-}_{\e}(k,s) \end{pmatrix} \|^{2} \\
&\lesssim r^{2}_{\theta}(\e) \int_{|k| \le K} dk ~ \int_{0}^{t} ds ~  \| \begin{pmatrix}\hat{\bar{F}}^{+}_{\e}(k,s) \\ \hat{\bar{F}}^{-}_{\e}(k,s) \end{pmatrix} \|^{2} \\
&= r^{2}_{\theta}(\e) \int_{|k| \le K} dk ~ \int_{0}^{t} ds ~  \| \begin{pmatrix}\hat{\bar{f}}^{+}_{\e}(k,s) \\ \hat{\bar{f}}^{-}_{\e}(k,s)\end{pmatrix} \|^{2} \\
&= 2 r^{2}_{\theta}(\e) \int_{|k| \le K} dk ~ \int_{0}^{t} ds ~ \| \begin{pmatrix}\hat{\bar{p}}_{\e}(k,s) \\ \hat{\bar{l}}_{\e}(k,s) \end{pmatrix} \|^{2} \\
&\lesssim r^{2}_{\theta}(\e) \int_{\T} dk ~ \int_{0}^{t} ds ~ \e \E_{\e} \Big[|\hat{\psi}(k,\frac{s}{n_{\theta}(\e)})|^{2} \Big] \lesssim r^{2}_{\theta}(\e).
\end{align}
\end{proof}

\subsection{Proof of Proposition \ref{prop:toshowfluc}}\label{subsec:proofofflucmean}

\begin{proof}

Thanks to Schwarz's inequality, we get
\begin{align}
&\int_{\frac{\T}{\e}} dk ~ \Big| \begin{pmatrix}\hat{\bar{F}}^{+}_{\e}(k,t) \\ \hat{\bar{F}}^{-}_{\e}(k,t) \end{pmatrix} - \begin{pmatrix}\tilde{\bar{F}}^{+}(k,t) \\ \tilde{\bar{F}}^{-}(k,t) \end{pmatrix} \Big|^{2} \\
&\le 2 \int_{\frac{\T}{\e}} dk ~ \Big| \begin{pmatrix}\hat{\bar{F}}^{+}_{\e}(k,t) \\ \hat{\bar{F}}^{-}_{\e}(k,t) \end{pmatrix} - \begin{pmatrix}\hat{\bar{F}}^{+,(0)}_{\e}(k,t) \\ \hat{\bar{F}}^{-,(0)}_{\e}(k,t) \end{pmatrix} \Big|^{2} + 2 \int_{\frac{\T}{\e}} dk ~ \Big| \begin{pmatrix}\hat{\bar{F}}^{+,(0)}_{\e}(k,t) \\ \hat{\bar{F}}^{-,(0)}_{\e}(k,t) \end{pmatrix} - \begin{pmatrix}\tilde{\bar{F}}^{+}(k,t) \\ \tilde{\bar{F}}^{-}(k,t) \end{pmatrix} \Big|^{2} \\
&=  2 \int_{\frac{\T}{\e}} dk ~ \Big| \Big[ \exp{\Big( M_{\e,\theta}(k) t \Big)} - \exp{\Big( M_{\theta}(k)t \Big)} \Big] \begin{pmatrix}\hat{\bar{F}}^{+}(k,0) \\ \hat{\bar{F}}^{-}(k,0) \end{pmatrix} \Big|^{2} \\
& \quad + 2 \int_{\frac{\T}{\e}} dk ~ \Big| \exp{\Big( M_{\theta}(k)t \Big)} \begin{pmatrix}\hat{\bar{F}}^{+}_{\e}(k,0) - \tilde{F}^{+}(k,0) \\ \hat{\bar{F}}^{-}_{\e}(k,0) - \tilde{F}^{-}(k,0) \end{pmatrix} \Big|^{2}. \label{ineq:toshowfluc1}
\end{align}
From \eqref{ass:phononic} and Lemma \ref{lem:boundofM}, we see that the second term of \eqref{ineq:toshowfluc1} vanishes as $\e \to 0$. Now we estimate the first term of \eqref{ineq:toshowfluc1}. From Lemma \ref{lem:convtomean}, it is sufficient to show that 
\begin{align}
\varlimsup_{K \to \infty} \varlimsup_{\e \to 0} \int_{K < |k| < \frac{1}{2 \e}} dk ~ \Big| \Big[ \exp{\Big( M_{\e,\theta}(k) t \Big)} - \exp{\Big( M_{\theta}(k)t \Big)} \Big] \begin{pmatrix}\hat{\bar{F}}^{+}(k,0) \\ \hat{\bar{F}}^{-}(k,0) \end{pmatrix} \Big|^{2} = 0.
\end{align}
By using \eqref{ass:phononic} and Lemma \ref{lem:boundofM} we have
\begin{align}
&\varlimsup_{\e \to 0} \int_{K < |k| < \frac{1}{2 \e}} dk ~ \Big| \Big[ \exp{\Big( M_{\e,\theta}(k) t \Big)} - \exp{\Big( M_{\theta}(k)t \Big)} \Big] \begin{pmatrix}\hat{\bar{F}}^{+}(k,0) \\ \hat{\bar{F}}^{-}(k,0) \end{pmatrix} \Big|^{2} \\
&\lesssim \varlimsup_{\e \to 0} \int_{K < |k| < \frac{1}{2 \e}} dk \Big| \begin{pmatrix}\hat{\bar{F}}^{+}(k,0) \\ \hat{\bar{F}}^{-}(k,0) \end{pmatrix} \Big|^{2} \\
&= \varlimsup_{\e \to 0} \int_{K < |k| < \frac{1}{2 \e}} dk \Big| \begin{pmatrix}\hat{\bar{f}}^{+}_{\e}(k,0) \\ \hat{\bar{f}}^{-}_{\e}(k,0)\end{pmatrix} \Big|^{2} \\
&= 2 \varlimsup_{\e \to 0} \int_{K < |k| < \frac{1}{2 \e}} dk ~ |\hat{\bar{p}}_{\e}(k,0)|^2 +  |\hat{\bar{l}}_{\e}(k,0)|^2 \\
&= 2 \int_{|\xi| > K} d\xi ~ |\tilde{\bar{p}}_{0}(\xi)|^2 + |\tilde{\bar{l}}_{0}(\xi)|^2.
\end{align}
Hence we complete the proof of this lemma. 

\end{proof}

\subsection{Proof of Theorem \ref{thm:flucnormalLLN}}\label{subsec:proofofflucnormalLLN}

Define 
\begin{align}
\hat{\mathcal{F}}^{\pm}_{\e}(k,t) := \hat{F}^{\pm}_{\e}(k,t) - \hat{\bar{F}}^{\pm}_{\e}(k,t).
\end{align}
We see that $|\hat{F}^{\pm}_{\e}(k,t)|^2 = |\hat{f}^{\pm}_{\e}(k,t)|^2$ and from \eqref{eq:energyconservation} we have
\begin{align}
\int_{\frac{\T}{\e}} dk ~ \mathbb{E}_{\e} \Big[ |\hat{f}^{+}_{\e}(k,t)|^2 +|\hat{f}^{-}_{\e}(k,t)|^2 \Big] &= 2 \int_{\frac{\T}{\e}} dk ~ \mathbb{E}_{\e} \Big[|\e \hat{p}(\e k,\frac{t}{n_{\theta}(\e)})|^2 + |\e \hat{l}(\e k,\frac{t}{n_{\theta}(\e)})|^2 \Big] \\
&= 2 \int_{\frac{\T}{\e}} dk ~ \mathbb{E}_{\e} \Big[|\e \hat{p}(\e k,0)|^2 + |\e \hat{l}(\e k,0)|^2 \Big]
\end{align}
for any $t \ge 0$. Thus we obtain
\begin{align}
&2 \int_{\frac{\T}{\e}} dk ~ \mathbb{E}_{\e} \Big[ |\e \hat{p}(\e k,0)|^2 + |\e \hat{l}(\e k,0)|^2 \Big] \\ 
&= \int_{\frac{\T}{\e}} dk ~ \mathbb{E}_{\e} \Big[ |\hat{\mathcal{F}}^{+}_{\e}(k,t)|^2 + |\hat{\mathcal{F}}^{-}_{\e}(k,t)|^2 \Big] + \int_{\frac{\T}{\e}} dk ~ |\hat{\bar{F}}^{+}_{\e}(k,t)|^2 + |\hat{\bar{F}}^{-}_{\e}(k,t)|^2.
\end{align}
By using \eqref{ass:phononic} and Proposition \ref{prop:toshowfluc} we get
\begin{align}
&\varlimsup_{\e \to 0}  \int_{\frac{\T}{\e}} dk ~ \mathbb{E}_{\e} \Big[ |\hat{\mathcal{F}}^{+}_{\e}(k,t)|^2 + |\hat{\mathcal{F}}^{-}_{\e}(k,t)|^2 \Big] \\ &= 2 \int_{\R} d\xi ~ |\tilde{\bar{p}}_{0}(\xi)|^2 + |\tilde{\bar{l}}_{0}(\xi)|^2 - \int_{\R} d\xi ~ |\tilde{\bar{F}}^{+}(\xi,t)|^2 + |\tilde{\bar{F}}^{-}(\xi,t)|^2.
\end{align}
If $2 < \theta < 4$ then for any $t \ge 0$ we have 
\begin{align}
|\tilde{\bar{F}}^{\pm}(\xi,t)|^2 = |\tilde{\bar{F}}^{\pm}(\xi,0)|^2 = |\tilde{\bar{p}}_{0}(\xi) \pm \tilde{\bar{l}}_{0}(\xi)|^2
\end{align}
and thus we obtain
\begin{align}
\varlimsup_{\e \to 0}  \int_{\frac{\T}{\e}} dk ~ \mathbb{E}_{\e} \Big[ |\hat{\mathcal{F}}^{+}_{\e}(k,t)|^2 + |\hat{\mathcal{F}}^{-}_{\e}(k,t)|^2 \Big] = 0. \label{lim:toshowflucLLN}
\end{align}
Proposition \ref{prop:toshowfluc} and \eqref{lim:toshowflucLLN} imply Theorem \ref{thm:flucnormalLLN} because
\begin{align}
&\varlimsup_{\e \to 0}\mathbb{E}_{\e}\Bigg[ \Big| \e \sum_{x \in \Z} f^{\pm}_{x}(\frac{t}{n_{\theta}(\e)})\Big(S^{\pm}_{\theta}(\frac{t}{m_{\theta}(\e)})J\Big)(\e x) - \int_{\R} dy ~ \bar{F}^{\pm}(y,t) J(y) \Big|  \Bigg] \\
&= \varlimsup_{\e \to 0}  \mathbb{E}_{\e}\Bigg[ \Big| \int_{\frac{\T}{\e}} dk ~ \Big( \hat{F}^{\pm}_{\e}(k,t) - \tilde{\bar{F}}^{\pm}(k,t) \Big) \tilde{J}(-k) \Big| \Bigg] \\
&= \varlimsup_{\e \to 0} \mathbb{E}_{\e}\Bigg[ \Big| \int_{\frac{\T}{\e}} dk ~ \hat{\mathcal{F}}^{\pm}_{\e}(k,t) \tilde{J}(-k) \Big| \Bigg]  \\
&\le  \|J\|_{\mathbb{L}^2(\R)} \varlimsup_{\e \to 0}\Bigg( \int_{\frac{\T}{\e}} dk ~ \mathbb{E}_{\e} \Big[ |\hat{\mathcal{F}}^{+}_{\e}(k,t)|^2 + |\hat{\mathcal{F}}^{-}_{\e}(k,t)|^2 \Big] \Bigg)^{\frac{1}{2}} = 0.
\end{align}

\section*{Acknowledgement}

We are grateful to Professor Stefano Olla for insightful discussions. HS was supported by JSPS KAKENHI Grant Number JP19J11268 and the Program for Leading Graduate Schools, MEXT, Japan. 

\appendix

\section{asymptotic behavior of $\hat{\a}(k), k \to 0$.}\label{app:asympofa}

\begin{lemma}\label{lem:asympofa}

\begin{align}
\hat{\a}(k) =& \begin{dcases} C_{1}(\theta) (2\pi)^{\theta - 1} |k|^{\theta - 1} + O(|k|^{\theta}) \quad &1 < \theta < 2, \\
C_{1}(2) 2 \pi |k| + O(k^{2} \log |k|^{-1}) \quad &\theta = 2, \\
C_{1}(\theta) (2\pi)^{\theta - 1} |k|^{\theta - 1} + C_{2}(\theta) (2\pi)^{2} |k|^{2} + O(|k|^{\theta}) \quad &2 < \theta < 3, \\ 
C_{1}(3) (2\pi)^{2} |k|^{2} \log |2\pi k|^{-1} + C_{2}(3) (2\pi)^{2} |k|^{2} + O(|k|^{3}) \quad &\theta = 3, \\ 
C_{1}(\theta) (2\pi)^{2} |k|^{2} + C_{2}(\theta) (2\pi)^{\theta - 1} |k|^{\theta - 1} + O(|k|^{\theta}) \quad &3 < \theta < 4, \\ 
C_{1}(\theta) (2\pi)^{2} |k|^{2} + C_{2}(\theta) (2\pi)^{\theta - 1} |k|^{\theta - 1} + O(|k|^{4}\log{|k|^{-1}}) \quad &\theta = 4, \\
C_{1}(\theta) (2\pi)^{2} |k|^{2} + C_{2}(\theta) (2\pi)^{\theta - 1} |k|^{\theta - 1} + O(|k|^{4}) \quad &4 < \theta < 5, \\
C_{1}(5) (2\pi)^{2} |k|^{2} + C_{2}(5) (2\pi)^{4} |k|^{4} \log |k|^{-1} + O(|k|^{4}) \quad &\theta = 5, \\
C_{1}(\theta) (2\pi)^{2} |k|^{2} + C_{2}(\theta) (2\pi)^{4} |k|^{4} + O(|k|^{\theta - 1}) \quad &5 < \theta < 7, \\
C_{1}(\theta) (2\pi)^{2} |k|^{2} + C_{2}(\theta) (2\pi)^{4} |k|^{4} + O(|k|^{6} \log{|k|^{-1}}) \quad & \theta = 7, \\
C_{1}(\theta) (2\pi)^{2} |k|^{2} + C_{2}(\theta) (2\pi)^{4} |k|^{4} + O(|k|^{6}) \quad & \theta > 7
\end{dcases} 
\end{align}
as $k \to 0$, where
\begin{align}
C_{1}(\theta) :=& \begin{dcases} \int_{0}^{\infty} dy ~ \frac{2- 2\cos{y}}{y^{\theta}} \quad &1 < \theta < 3, \\ 1 \quad &\theta = 3, \\ \sum_{x \ge 1} |x|^{-\theta} \quad &\theta > 3, \end{dcases} 
\end{align}
and
\begin{align}
C_{2}(\theta) :=& \begin{dcases} - \int_{0}^{1} dy ~ \frac{1}{y^{\theta - 2}} + \int_{1}^{\infty} dy ~ ( [y]^{-(\theta - 2)} - y^{-(\theta - 2)}) \quad &2 < \theta < 3, \\
\frac{3}{2} \quad &\theta = 3, \\ 
\int_{0}^{\infty} dy ~ \frac{2 - 2\cos{y} - y^{2}}{y^{\theta}} \quad &3 < \theta < 5, \\
- \frac{1}{12} \quad &\theta = 5, \\
- \frac{1}{12} \sum_{x \ge 1} \frac{1}{x^{\theta - 4}} \quad &\theta > 5. \end{dcases} 
\end{align} 
Moreover, $[y], y \in \R$ is the greatest integer less than or equal to $y$. 
\end{lemma}

\begin{proof}

First we observe that $\hat{\a}$ can be devided into three parts: 
\begin{align}
\hat{\a}(k) &= \sum_{x \in \Z} \a_{x} e^{-2 \pi k x} = 2 \int_{1}^{\infty} dy ~ \frac{1 - \cos{2\pi k [y]}}{[y]^{\theta}} = \a_{1}(k) + \a_{2}(k) + \a_{3}(k) ,
\end{align}
where
\begin{align}
\a_{1}(k) &:= \begin{dcases} (2\pi)^{\theta - 1} |k|^{\theta - 1} \int_{0}^{\infty} dy ~ \frac{2 - 2 \cos{y}}{y^{\theta}} \quad &1 < \theta < 3, \\
(2\pi)^{2} |k|^{2} \log{|2\pi k|^{-1}} \quad & \theta = 3, \\
(2\pi)^{2} k^{2} \sum_{x \ge 1} \frac{1}{|x|^{\theta - 2}} \quad & \theta > 3,
\end{dcases}\\
\a_{2}(k) &:= \begin{dcases} 0 \quad &1 < \theta \le 2, \\ 
- (2\pi)^{\theta - 1} |k|^{\theta - 1} \int_{0}^{2 \pi |k|} dy ~ \frac{1}{y^{\theta - 2}} + (2\pi)^{2} |k|^{2} \int_{1}^{\infty} dy ~ \frac{1}{[y]^{\theta-2}} - \frac{1}{y^{\theta - 2}} \quad &2 < \theta < 3, \\
(2\pi)^{2} |k|^{2} \bigl\{ \int_{1}^{\infty} dy ~ \frac{1}{[y]} - \frac{1}{y}  & \\ \quad + \int_{1}^{\infty} dy ~ \frac{2 - 2 \cos{y}}{y^{3}} + \int_{0}^{1} dy ~ \frac{2 - 2 \cos{y} - y^{2}}{y^{3}} \bigr\} \quad &\theta = 3, \\
(2\pi)^{\theta - 1} |k|^{\theta - 1} \int_{0}^{\infty} dy ~ \frac{2 - 2 \cos^{2}{y} - y^{2} }{y^{\theta}} \quad &3 < \theta < 5, \\
- \frac{(2\pi)^{4} |k|^{4} \log{|k|^{-1}}}{12} \quad &\theta = 5, \\
- \frac{(2\pi)^{4} |k|^{4}}{12} \sum_{x \ge 1} \frac{1}{x^{\theta - 4}} \quad &\theta > 5,
\end{dcases} \\
\a_{3}(k) &:= \begin{dcases} \int_{1}^{\infty} dy ~ \frac{2 - 2 \cos{2 \pi k [y]}}{[y]^{\theta}} - \frac{2 - \cos{2 \pi k y}}{y^{\theta}} & \\ \quad - (2\pi)^{\theta - 1} |k|^{\theta - 1} \int_{0}^{2 \pi |k|} dy ~ \frac{2 - 2 \cos{y}}{y^{\theta}}  \quad &1 < \theta \le 2, \\ 
\int_{1}^{\infty} dy ~ \frac{2 - 2 \cos{2\pi k [y]} - (2\pi)^{2} k^{2} [y]^{2} }{[y]^{\theta}} - \frac{2 - 2 \cos{2\pi k y} - (2\pi)^{2} k^{2} y^{2} }{y^{\theta}} \quad & \\ 
\quad - (2\pi)^{\theta - 1} |k|^{\theta - 1} \int_{0}^{2 \pi |k|} dy ~ \frac{2 - 2 \cos{y} - y^{2}}{y^{\theta}} \quad &2 < \theta < 5, \\ 
\int_{1}^{\infty} dy ~ \frac{2 - 2 \cos{2\pi k [y]} - (2\pi)^{2} k^{2} [y]^{2} }{[y]^{5}} - \frac{2 - 2 \cos{2\pi k y} - (2\pi)^{2} k^{2} y^{2} }{y^{5}} \quad & \\ \quad + (2\pi)^{4}|k|^{4} \int_{1}^{\infty} dy ~ \frac{2 - 2 \cos{y} - y^{2} }{y^{5}} & \\  \quad + \frac{(2\pi)^{4} |k|^{4}}{12} \int_{2 \pi |k|}^{1} dy ~ \frac{24 - 24 \cos{y} - 12 y^{2} + y^{4}}{y^{5}} + \frac{(2\pi)^{4} |k|^{4} \log{2 \pi}}{12}  \quad &\theta = 5, \\
\frac{1}{12} \sum_{x \ge 1} \frac{24 - 24 \cos{2 \pi k x} - 12 (2\pi)^{2} k^{2} x^{2} + (2\pi)^{4} k^{4} x^{4}}{x^{\theta}} \quad &\theta > 5.
\end{dcases}
\end{align}
Each $\a_{i}$ corresponds to $i$-th order term of $\hat{\a}$ and we can easily see the asymptotic behavior of $\a_{1}, \a_{2}$. In the following subsections we compute the order of the remainder term $\a_{3}$. Since in the case $\theta > 5$ we can repeat almost the same argument, we only compute the dcases $1 < \theta \le 2$, $2 < \theta < 5$, $\theta = 5$. Before showing the asymptotic behavior of $\a_{3}$, we note that if $\theta = 3$ then $\a_{2}(k) = \frac{3}{2} k^{2}$ because
\begin{align}
\int_{1}^{\infty} dy ~ \frac{1}{[y]} - \frac{1}{y} &= \gamma_{E}, \\
\int_{1}^{\infty} dy ~ \frac{2 - 2 \cos{y}}{y^{3}} &= 1 - \cos{1} + \int_{1}^{\infty} dy ~ \frac{\sin{y}}{y^{2}} \\
&= 1 - \cos{1} + \sin{1} - \operatorname{Ci}(1), \\
\int_{0}^{1} dy ~ \frac{2 - 2 \cos{y} - y^{2}}{y^{3}} &= -\frac{1}{2} + \cos{1} + \int_{0}^{1} dy ~ \frac{\sin{y} - y}{y^{2}} \\
&= \frac{1}{2} + \cos{1} - \sin{1} + \int_{0}^{1} dy ~ \frac{\cos{y} - 1}{y} \\
&= \frac{1}{2} + \cos{1} - \sin{1} - \gamma_{E} + \operatorname{Ci}(1),
\end{align}
where $\gamma_{E}$ is the Euler's constant and $\operatorname{Ci}$ is the cosine integral. At the last line we use the following equation
\begin{align}
\operatorname{Ci}(x) := - \int_{x}^{\infty} dy \frac{\cos{y}}{y} = \gamma_{E} + \log{x} +  \int_{0}^{x} dy ~ \frac{\cos{y} - 1}{y}.
\end{align}

\subsection{When $1 < \theta \le 2$.}

Define $f_{1}(y) := \frac{1 - \cos{2 \pi k y}}{y^{\theta}} , y > 0$. Since
\begin{align}
f_{1}([y]) = f_{1}(y) + f_{1}^{'}(y)([y] - y) + \frac{1}{2}f_{1}^{''}(y^{'})([y] - y)^{2}
\end{align}
for some $[y] \le y^{'} \le y$, and 
\begin{align}
|f_{1}^{'}(y)| \lesssim |k| \frac{|\sin{2 \pi k y}|}{y^{\theta}}, \quad |f_{1}^{''}(y^{'})| \lesssim |k|^{2} \frac{1}{[y]^{\theta}},
\end{align}
we have
\begin{align}
&\bigl| \int_{1}^{\infty} dy ~ f_{1}([y]) - f_{1}(y) \bigr| \lesssim |k| \int_{1}^{\infty} dy ~ \frac{|\sin{ 2\pi k y}|}{y^{\theta}} + |k|^{2} \int_{1}^{\infty} dy ~ \frac{1}{[y]^{\theta}} \\
& \quad \lesssim |k|^{\theta} \bigl\{ \int_{|k|}^{1} dy ~ \frac{|\sin{2 \pi y}|}{y^{\theta}} + \int_{1}^{\infty} dy ~ \frac{|\sin{2 \pi y}|}{y^{\theta}} \bigr\} + |k|^{2} \\
& \quad \lesssim \begin{dcases} |k|^{\theta} \quad & 1 < \theta < 2, \\ |k|^{2} \log{|k|^{-1}} \quad & \theta = 2. \end{dcases}
\end{align}
In addition, we get
\begin{align}
|k|^{\theta - 1} \bigl| \int_{0}^{2 \pi |k|} dy ~ \frac{1-\cos{y}}{y^{\theta}} \bigr| \lesssim |k|^{\theta - 1}  \int_{0}^{|k|} dy ~ \frac{1}{y^{\theta - 2}} \lesssim |k|^{2}.
\end{align}

\subsection{When $2 < \theta < 5$.}

Define $f_{2}(y) := \frac{2 - 2 \cos{2\pi k y} - (2\pi)^{2}k^{2}y^{2}}{y^{\theta}} , y > 0$. Since
\begin{align}
f_{2}([y]) = f_{2}(y) + f_{2}^{'}(y)([y] - y) + \frac{1}{2}f_{2}^{''}(y^{'})([y] - y)^{2}
\end{align}
for some $[y] \le y^{'} \le y$, we have
\begin{align}
|f_{2}^{'}(y)| &\lesssim |k| \frac{|\sin{2\pi k y} - 2 \pi k y|}{y^{\theta}}, \\ 
|f_{2}^{''}(y^{'})| &\lesssim |k|^{2} \bigl\{ \frac{1 - \cos{2 \pi k y^{'}}}{(y^{'})^{\theta}} + \frac{ 2\pi |k| y^{'} - \sin{2\pi |k| y^{'}} }{(y^{'})^{\theta + 1}} + \frac{ \pi^{2}k^{2}(y^{'})^{2} - \sin^{2}{\pi k y^{'}}}{(y^{'})^{\theta + 2}} \bigr\} \\
&\lesssim \begin{dcases} |k|^{3} [y]^{-\theta} \quad &2 < \theta \le 3, \\ |k|^{4} [y]^{2-\theta} \quad &\theta > 3, \end{dcases} 
\end{align}
and 
\begin{align}
\int_{1}^{\infty} dy ~ |f_{2}^{'}(y)| &\lesssim \begin{dcases} |k|^{\theta} \bigl\{  \int_{2 \pi |k|}^{1} dy ~ \frac{y - \sin{y}}{y^{\theta}} + \int_{1}^{\infty} dy ~ \frac{y - \sin{y}}{y^{\theta}} \bigr\} \quad &2 < \theta \le 4, \\ 
|k|^{4} \int_{1}^{\infty} dy ~ \frac{1}{y^{\theta - 3}} \quad & 4 < \theta < 5, \end{dcases} \\
&\lesssim \begin{dcases} |k|^{\theta} \quad &2 < \theta < 4, \\ |k|^{4} \log{|k|^{-1}} \quad &\theta = 4, \\ |k|^{4} \quad &\theta > 4. \end{dcases}
\end{align}
In addition, we get
\begin{align}
|k|^{\theta - 1} \bigl| \int_{0}^{2\pi |k|} dy ~ \frac{2 - 2 \cos{y} - y^{2}}{y^{\theta}} \bigr| \lesssim |k|^{\theta - 1}  \int_{0}^{2\pi|k|} dy ~ \frac{1}{y^{\theta - 4}} \lesssim |k|^{4}.
\end{align}

\subsection{When $\theta = 5$.}

Since
\begin{align}
|k|^{4} \bigl| \int_{2 \pi |k|}^{1} dy ~ \frac{24 - 24 \cos{y} - 12 y^{2} + y^{4}}{y^{5}} \bigr| \lesssim |k|^{6},
\end{align}
we can see that $\a_{3}(k) = O(|k|^{4})$.

\end{proof}

\begin{lemma}\label{lem:asympofaco}

\begin{align}
&\hat{\a}(k+k^{'}) - \hat{\a}(k) - \hat{\a}(k^{'}) \\ 
&= \begin{dcases} (2 \pi)^{\theta - 1}C_{1}(\theta) \Big(|k + k^{'}|^{\theta - 1} - |k|^{\theta - 1} - |k^{'}|^{\theta - 1} \Big) + O( |k|^{\theta - 1} |k^{'}| + |k| |k^{'}|^{\theta - 1} )  & 1 < \theta < 2, \\
2 \pi C_{1}(2) \Big(|k + k^{'}| - |k| - |k^{'}| \Big) + O(|k| |k^{'}| \log{|k|^{-1}} + |k| |k^{'}| \log{|k^{'}|^{-1}}) & \theta = 2, \\
(2 \pi)^{\theta - 1}C_{1}(\theta) \Big(|k + k^{'}|^{\theta - 1} - |k|^{\theta - 1} - |k^{'}|^{\theta - 1} \Big) + O(|kk^{'}|) & 2 < \theta < 3, \\
(2\pi)^2 \Big( |k + k^{'}|^2 \log{\big(|k + k^{'}|^{-1} \big)} - |k|^2 \log{\big(|k|^{-1} \big)} -  |k^{'}|^2 \log{\big(|k^{'}|^{-1} \big)}  \Big) + O(|kk^{'}|) & \theta = 3, \\
2 (2 \pi)^2 C_{1}(\theta) k k^{'} + o(|kk^{'}|) & \theta > 3, \end{dcases}
\end{align}
for any $k,k^{'} \in \T$. In addition, if $\theta = 3$ then we obtain
\begin{align}
&|\hat{\a}(k+k^{'}) - \hat{\a}(k) - \hat{\a}(k^{'}) | \\
&\le 2 (2\pi)^2 |kk^{'}| \sqrt{\log(|k|^{-1}) \log(|k^{'}|^{-1})} + O(|kk^{'}|\sqrt{\log(|k|^{-1})} + |kk^{'}|\sqrt{\log(|k^{'}|^{-1})}).
\end{align}
\end{lemma}

\begin{proof}

First we observe that
\begin{align}
\hat{\a}(k+k^{'}) - \a(k) - \a(k^{'}) &= 2 \sum_{x \ge 1} \frac{\cos{(2\pi k x)} + \cos{(2\pi k^{'} x)} - \cos{\big(2\pi (k+k^{'}) x \big)} -1}{|x|^{\theta}} \\ &= 2 \sum_{x \ge 1} \frac{\sin{(2\pi k x)} \sin{(2\pi k^{'} x)} - \Big(1 - \cos{(2\pi k x)}\Big) \Big(1 - \cos{(2\pi k^{'} x)}\Big)}{|x|^{\theta}}.
\end{align}

\subsection{When $1 < \theta < 3$}

Define 
\begin{align}
f_{3}(y) := 2 \frac{\sin{(2\pi k y)} \sin{(2\pi k^{'} y)} - \Big(1 - \cos{(2\pi k y)}\Big) \Big(1 - \cos{(2\pi k^{'} y)}\Big)}{|y|^{\theta}}.
\end{align}
Then we have
\begin{align}
\a(k) + \a(k^{'}) - \hat{\a}(k+k^{'}) = \int_{1}^{\infty} dy ~ f_{3}(y) +  \int_{1}^{\infty} dy ~ f_{3}([y]) - f_{3}(y).
\end{align}
For any  $y \ge 1$, there exists some $y^{'} \ge 1, [y] \le y^{'} \le y$ such that
\begin{align}
f_{3}([y]) = f_{3}(y) + f_{3}^{'}(y^{'})([y] - y),
\end{align}
and we have
\begin{align}
|f_{3}^{'}(y^{'})| &\lesssim |k| \frac{|\sin{2 \pi k^{'}y}|}{|y|^{\theta}} + |k^{'}| \frac{|\sin{2 \pi k y}|}{|y|^{\theta}} 
\end{align}
where we use an inequality $[y]^{-1} \le 2 y^{-1}$ on $\{y \ge 1\}$. Since
\begin{align}
\int_{1}^{\infty} dy ~ \frac{|\sin{2 \pi k y}|}{|y|^{\theta}} &\lesssim \begin{dcases} |k|^{\theta - 1} \quad & 1 < \theta < 2, \\
|k| \log{|k|^{-1}} \quad & \theta = 2, \\
|k| \quad & \theta > 2,
\end{dcases} 
\end{align}
and
\begin{align}
\int_{1}^{\infty} dy ~ f_{3}(y) &= \int_{1}^{\infty} dy ~ \frac{2 -  2\cos{\big(2\pi (k+k^{'}) x \big)} }{|y|^{\theta}} \\ 
& \quad - \int_{1}^{\infty} dy ~ \frac{2 - 2 \cos{(2\pi k x)}}{|y|^{\theta}} - \int_{1}^{\infty} dy ~ \frac{2 - 2 \cos{(2\pi k^{'} x)}}{|y|^{\theta}} \\
&= (2 \pi)^{\theta - 1}C_{1}(\theta) \Big( |k + k^{'}|^{\theta - 1} - |k|^{\theta - 1} - |k^{'}|^{\theta - 1} \Big) \\ 
&\quad - 2 \int_{0}^{1} dy ~ \frac{\sin{(2\pi k y)} \sin{(2\pi k^{'} y)} - \Big(1 - \cos{(2\pi k y)}\Big) \Big(1 - \cos{(2\pi k^{'} y)}\Big)}{|y|^{\theta}} \\
&= (2 \pi)^{\theta - 1}C_{1}(\theta) \Big(|k + k^{'}|^{\theta - 1} - |k|^{\theta - 1} - |k^{'}|^{\theta - 1} \Big) + O(|k k^{'}|),
\end{align}
we have
\begin{align}
\hat{\a}(k+k^{'}) - \a(k) - \a(k^{'}) &= (2 \pi)^{\theta - 1}C_{1}(\theta) \Big(|k + k^{'}|^{\theta - 1} - |k|^{\theta - 1} - |k^{'}|^{\theta - 1} \Big) \\
& \quad + \begin{dcases} O( |k|^{\theta - 1} |k^{'}| + |k| |k^{'}|^{\theta - 1} ) \quad & 1 < \theta < 2, \\
O(|kk^{'}| \log{|k|^{-1}} + |kk^{'}| \log{|k^{'}|^{-1}})\quad & \theta = 2, \\
O(|kk^{'}|) \quad & \theta > 2. \end{dcases}
\end{align}

\subsection{When $\theta = 3$}

In this case, from the following decomposition, 
\begin{align}
\int_{1}^{\infty} dy ~ f_{3}(y) &= (2\pi)^2 |k + k^{'}|^2 \log{\big(|2\pi (k + k^{'})|^{-1} \big)}  \\ & \quad -  (2\pi)^2 |k|^2 \log{\big(|2\pi k|^{-1} \big)} -  (2\pi)^2 |k^{'}|^2 \log{\big(|2\pi k^{'}|^{-1} \big)}  \\
& \quad + (2\pi)^2 \Big( |k + k^{'}|^2 -  |k|^2 -  |k^{'}|^2 \Big) \int_{1}^{\infty} dy ~ \frac{2 - 2\cos{(y)}}{|y|^{3}} \\
& \quad + (2\pi)^2 \Big( |k + k^{'}|^2 -  |k|^2 -  |k^{'}|^2 \Big) \int_{0}^{1} dy \frac{2 - 2\cos{(y)} - y^2}{|y|^{3}} \\
& \quad - 2 \int_{0}^{1} dy ~ \frac{\sin{(2\pi k y)} \sin{(2\pi k^{'} y)} - 4 \pi^2 k k^{'} y^2}{|y|^3}  \\
& \quad + 2 \int_{0}^{1} dy ~ \frac{\Big(1 - \cos{(2\pi k y)}\Big) \Big(1 - \cos{(2\pi k^{'} y)}\Big)}{|y|^3},
\end{align}
we obtain
\begin{align}
&\hat{\a}(k+k^{'}) - \a(k) - \a(k^{'}) \\
&= (2\pi)^2 \Big( |k + k^{'}|^2 \log{\big(|k + k^{'}|^{-1} \big)} - |k|^2 \log{\big(|k|^{-1} \big)} -  |k^{'}|^2 \log{\big(|k^{'}|^{-1} \big)}  \Big) + O(|kk^{'}|).
\end{align}

Note that by using Schwarz's inequality we obtain the following boundedness : 
\begin{align}
\Big| \int_{1}^{\infty} dy ~ f_{3}(y) \Big| &= 2 \Big| \int_{1}^{\infty} \frac{\sin{(2\pi k y)} \sin{(2\pi k^{'} y)} - \Big(1 - \cos{(2\pi k y)}\Big) \Big(1 - \cos{(2\pi k^{'} y)}\Big)}{|y|^{3}} \Big| \\
&= 2 \Big| \int_{1}^{\infty} \frac{\sin{(2\pi k y)} \sin{(2\pi k^{'} y)} }{|y|^{3}} \Big| + O(|kk^{'}|) \\
&\le 2  \Big( \int_{1}^{\infty} \frac{|\sin{(2\pi k y)} |^2}{|y|^{3}} \Big)^{\frac{1}{2}}\Big( \int_{1}^{\infty} \frac{|\sin{(2\pi k^{'} y)} |^2}{|y|^{3}} \Big)^{\frac{1}{2}} + O(|kk^{'}|) \\
&= 2\Big( (2\pi)^2 |k|^{2} \log(|k|^{-1}) + O(|k^2|)  \Big)^{\frac{1}{2}} \Big( (2\pi)^2 |k^{'}|^{2} \log(|k^{'}|^{-1}) + O(|k^2|)  \Big)^{\frac{1}{2}} + O(|kk^{'}|) \\
&= 2 (2\pi)^2 |kk^{'}| \sqrt{\log(|k|^{-1}) \log(|k^{'}|^{-1})} + O(|kk^{'}|\sqrt{\log(|k|^{-1})} + |kk^{'}|\sqrt{\log(|k^{'}|^{-1})}).
\end{align}

\subsection{When $\theta > 3$}

In this case we can easily see the order because
\begin{align}
&2 \sum_{x \ge 1} \frac{\sin{(2\pi k x)} \sin{(2\pi k^{'} x)} - \Big(1 - \cos{(2\pi k x)}\Big) \Big(1 - \cos{(2\pi k^{'} x)}\Big)}{|x|^{\theta}} \\
&= 2 (2 \pi)^2 C_{1}(\theta) k k^{'} + 2 \sum_{x \ge 1} \frac{\sin{(2\pi k x)} \sin{(2\pi k^{'} x)} - (2 \pi)^2 k k^{'} x^2}{|x|^{\theta}} \\
& \quad - 2 \sum_{x \ge 1} \frac{\Big(1 - \cos{(2\pi k x)}\Big) \Big(1 - \cos{(2\pi k^{'} x)}\Big)}{|x|^{\theta}} \\
&= 2 (2 \pi)^2 C_{1}(\theta) k k^{'} + o(|kk^{'}|).
\end{align}

\end{proof}

\section{On the equivalence of (\ref{ass:phononic}) and (\ref{ass:equiassofthm3})}\label{app:initialequi}

First we observe that 
\begin{align}  
&\e \sum_{x \in \Z} |p_{x} - \bar{p}_{0}(\e x)|^{2} \\ &\le  \e \sum_{x \in \Z} \Big|p_{x} - \int_{\frac{\T}{\e}} dk ~ e^{2 \pi \sqrt{-1} k \e x} \tilde{\bar{p}}_{0}(k) \Big|^{2} +  \e \sum_{x \in \Z} \Big|\bar{p}_{0}(\e x) -  \int_{\frac{\T}{\e}} dk ~ e^{2 \pi \sqrt{-1} k \e x} \tilde{\bar{p}}_{0}(k) \Big|^{2} \\
&= \int_{\frac{\T}{\e}} dk ~ |\e \hat{p}(\e k) - \tilde{\bar{p}}_{0}(k)|^{2} + \e \sum_{x \in \Z} \Big|\bar{p}_{0}(\e x) -  \int_{\frac{\T}{\e}} dk ~ e^{2 \pi \sqrt{-1} k \e x} \tilde{\bar{p}}_{0}(k) \Big|^{2}
\end{align}
and
\begin{align}
&\int_{\frac{\T}{\e}} dk ~ |\e \hat{p}(\e k) - \tilde{\bar{p}}_{0}(k)|^{2} = \e \sum_{x \in \Z} |p_{x} - \frac{1}{\e} \int_{\T} dk ~ e^{2 \pi \sqrt{-1} k x} \tilde{\bar{p}}_{0}(\frac{k}{\e})|^2 \\
&\le \e \sum_{x \in \Z} |p_{x} -  \bar{p}_{0}(\e x)|^2 + \e \sum_{x \in \Z} \Big|\bar{p}_{0}(\e x) - \int_{\frac{\T}{\e}} dk ~ e^{2 \pi \sqrt{-1} k \e x} \tilde{\bar{p}}_{0}(k) \Big|^2,
\end{align}
where we use Parseval's identity. Therefore to check the equivalence of \eqref{ass:phononic} and \eqref{ass:equiassofthm3}, it is sufficient to show that 
\begin{align}
\lim_{\e \to 0} \e \sum_{x \in \Z} |\bar{p}_{0}(\e x) - \int_{\frac{\T}{\e}} dk ~ e^{2 \pi \sqrt{-1} k \e x} \tilde{\bar{p}}_{0}(k) |^2 = 0. \label{lim:equisuffi}
\end{align}
Since $\bar{p}_0 \in C^{\infty}_{0}(\R)$, we have
\begin{align}
&\Big| \bar{p}_{0}(\e x) - \int_{\frac{\T}{\e}} dk ~ e^{2 \pi \sqrt{-1} k \e x} \tilde{\bar{p}}_{0}(k) \Big| = \Big| \int_{|k| > \frac{1}{2\e}} dk ~ e^{2 \pi \sqrt{-1} k \e x} \tilde{\bar{p}}_{0}(k) \Big| \\
&= \Bigg| \frac{1}{2 \pi \sqrt{-1} \e x} \Bigg( \Big( e^{ - \pi \sqrt{-1} x} \tilde{\bar{p}}_{0}(- \frac{1}{2\e}) - e^{\pi \sqrt{-1} x} \tilde{\bar{p}}_{0}(\frac{1}{2\e})  \Big)  -  \int_{|k| > \frac{1}{2\e}} dk ~ e^{2 \pi \sqrt{-1} k \e x} \partial_{k} \tilde{\bar{p}}_{0}(k) \Bigg) \Bigg| \\
&\lesssim \frac{\e}{|x|}, 
\end{align}
and thus we obtain \eqref{lim:equisuffi}.

\end{document}